\documentclass[preprintnumbers,amsmath,amssymb,superscriptaddress]{revtex4}
\usepackage{amsmath,amssymb,braket,bm}
\usepackage{times}
\usepackage{graphicx}
\usepackage{dcolumn}
\usepackage{dsfont}
\usepackage{mathrsfs}
\usepackage[landscape,papersize={297.1mm,210mm},left=1.9cm,right=1.6cm,top=2.7cm,bottom=2.8cm]{geometry}
\usepackage{amsthm}
\usepackage{xcolor}
\usepackage{hyperref}
\usepackage{float} 
\makeatletter

\newcommand{\Rmnum}[1]{\expandafter\@slowromancap\romannumeral #1@}

\newtheorem{theorem}{Theorem}
\makeatother

\begin{document}
	
\title{Discrimination of genuinely nonlocal sets without entanglement in multipartite systems}

\author{Ziying Hou}
\affiliation{School of Mathematics and Science, Hebei GEO University, Shijiazhuang 052161, China}

\author{Huaqi Zhou}
\email{zhouhuaqilc@163.com}
\affiliation{School of Mathematics and Science, Hebei GEO University, Shijiazhuang 052161, China}
\affiliation{Hebei Province Key Laboratory of Intelligent Sensing and Data Processing for Geo-environment, Hebei GEO University, Shijiazhuang 052161, China}
\author{Limin Gao}
\email{gaoliminabc@163.com}
\affiliation{School of Mathematics and Science, Hebei GEO University, Shijiazhuang 052161, China}
\affiliation{Hebei Province Key Laboratory of Intelligent Sensing and Data Processing for Geo-environment, Hebei GEO University, Shijiazhuang 052161, China}

\begin{abstract}
Genuine nonlocality arises when a set of multipartite orthogonal states is locally indistinguishable under any bipartition of the subsystems. The entanglement-assisted discrimination of such genuinely nonlocal orthogonal product sets has attracted significant attention in quantum information. Based on the criterion of local irreducibility, genuine nonlocality is classified into Type I (reducible) and Type II (irreducible). We present entanglement-assisted discrimination schemes for both types of genuinely nonlocal sets that use minimal resources. For low-dimensional cases, Type I sets require only a single EPR pair, whereas Type II sets necessitate only one GHZ state. We extend these protocols to higher-dimensional systems: the discrimination of Type I sets requires only one maximally entangled state in a two-qutrit system, while that of Type II sets similarly demands a single maximally entangled state in a three-qutrit system. For $n$-partite ($n > 3$) systems, Type I sets continue to require only one maximally entangled state, whereas Type II sets necessitate just one additional EPR pair compared to their Type I counterparts. These results provide a robust framework for the efficient discrimination of genuinely nonlocal sets using minimal quantum resources.
\end{abstract}
\maketitle
	
\section{INTRODUCTION}

Quantum nonlocality is a central topic in quantum information theory. A set of orthogonal quantum states is locally indistinguishable if it cannot be perfectly distinguished by any sequence of local operations and classical communication (LOCC). This phenomenon is distinct from Bell nonlocality~\cite{brunner2014bell}, which originates from entangled pure states~\cite{gao2014permutationally, Horodecki}. In 1999, Bennett \textit{et al.}~\cite{bennett1999quantum} first constructed a set of orthogonal product states in $\mathbb{C}^3 \otimes \mathbb{C}^3$ that cannot be distinguished via LOCC, demonstrating that nonlocality can emerge in the absence of entanglement. Subsequently, significant progress has been made in nonlocality theory and the construction of such state sets~\cite{ZGY22, ZGY23, HGY24, HGY25, ZGY25, cao2023locally, walgate2002nonlocality, wang2017local, xu2015detecting, xu2016locally, xu2021novel, zhang2016local, zhang2017construction, zuo2021nonlocal}. Additionally, genuine nonlocality in multipartite systems has been extensively studied~\cite{rout2019genuinely}.

A set of multipartite orthogonal quantum states is genuinely nonlocal if it is locally indistinguishable under every bipartition. If it is impossible to eliminate one or more quantum states from the set using nontrivial orthogonality-preserving local measurements (OPLMs), the set is defined as locally irreducible. In 2019, Rout \textit{et al.}~\cite{rout2019genuinely} introduced the formal definition and classification of genuinely nonlocal product bases. They categorized the genuine nonlocality of orthogonal state sets into two types: a genuinely nonlocal set has \textit{type I genuine nonlocality} if it is locally reducible under the full separation of all subsystems; otherwise, it possesses \textit{type II genuine nonlocality}. In 2023, Xiong \textit{et al.}~\cite{xiong2023distinguishability} investigated distinguishability-based genuine nonlocality, using GHZ states as an example, and constructed genuinely nonlocal sets in tripartite systems with cardinality as low as $d+3$, where $d = 2^t$ for $t \geq 1$. In 2024, Lu \textit{et al.}~\cite{Lu2024GenuinelyNS} constructed a locally indistinguishable orthogonal product set (OPS) of size $2d_2-1$ in $\mathbb{C}^{d_1} \otimes \mathbb{C}^{d_2}$ ($3 \le d_1 \le d_2$), which improves upon existing results. Based on this bipartite construction, genuinely nonlocal OPSs of type I were proposed for tripartite systems $\mathbb{C}^{d_1} \otimes \mathbb{C}^{d_2} \otimes \mathbb{C}^{d_3}$ ($3 \le (d_1-1) \le d_2 \le d_3$) and subsequently extended to $n$-partite systems. Furthermore, type II genuinely nonlocal sets without entanglement were constructed in $\mathbb{C}^{d_1} \otimes \mathbb{C}^{d_2} \otimes \mathbb{C}^{d_3}$ ($3 \le d_1 \le d_2 \le d_3$) and $\bigotimes_{i=1}^n \mathbb{C}^{d_i}$ ($3 \le d_1 \le d_2 \le \cdots \le d_n$), respectively. Extensive research advances have since been achieved regarding the construction of genuinely nonlocal multipartite orthogonal states~\cite{rout2021multiparty, li2021local, xiong2024small}.

For a given genuinely nonlocal set, it is impossible to precisely identify every quantum state solely via LOCC. When information is encoded in such sets, global measurements are typically required to extract it. Under the restriction of LOCC, even orthogonal product sets may become indistinguishable~\cite{feng2009characterizing, yang2015characterizing, halder2018several, xu2016quantum}. This property confers significant practical value in fields such as quantum secret sharing and quantum data hiding~\cite{wang2017quantum, yang2015quantum, guo2001quantum, divincenzo2002quantum, jiang2019arbitrary, jiang2018multiparty, gao2005deterministic}. However, when sufficient entanglement is shared, such nonlocal sets can be perfectly distinguished~\cite{cohen2007local, bandyopadhyay2018optimal, li2019local}. Entanglement is an extremely valuable resource~\cite{gao2010detection, yan2011two}. Rout \textit{et al.}~\cite{rout2019genuinely} proposed local discrimination protocols assisted by various entanglement resources, demonstrating higher entanglement efficiency compared with conventional teleportation-based schemes. Developing LOCC discrimination schemes that utilize low-dimensional entanglement and minimize resource consumption remains a subject of significant research interest~\cite{ghosh2001distinguishability, zhang2016entanglement, gungor2016entanglement, bandyopadhyay2016entanglement, mikova2014optimal, shi2020unextendible}.

In this work, We study the entanglement-assisted local distinguishability of genuinely nonlocal orthogonal product sets. In Sec. \ref{Q2}, we introduce the necessary preliminaries. In Sec. \ref{Q3}, we present entanglement-assisted discrimination schemes for known locally indistinguishable sets in the $\mathbb{C}^{d_1} \otimes \mathbb{C}^{d_2}$ system. In Sec. \ref{Q4}, we propose entanglement-assisted discrimination schemes for two genuinely nonlocal type I OPSs~\cite{Lu2024GenuinelyNS} and generalize these schemes to $\bigotimes_{i=1}^n \mathbb{C}^{d_i}$ systems. In Sec. \ref{Q5}, we present entanglement-assisted discrimination schemes for two genuinely nonlocal type II OPSs and extend them to $\bigotimes_{i=1}^n \mathbb{C}^{d_i}$ systems. Finally, in Sec.~\ref{Q6}, we provide a brief summary.
	
\section{PRELIMINARIES}\label{Q2}
In this section, we introduce the notation and basic concepts used throughout this paper.

\textbf{Definition 1.} \cite{bennett1999quantum}
A set of orthogonal states is locally indistinguishable or nonlocal if it cannot be perfectly distinguished under local operations and classical communications (LOCC).

\textbf{Definition 2.} \cite{walgate2002nonlocality}
A measurement is trivial if all the POVM elements are proportional to the identity operator; otherwise, the measurement is nontrivial.

\textbf{Definition 3.} \cite{halder2019strong}
A measurement of a set of orthogonal quantum states is called an orthogonality-preserving local measurement (OPLM) if the post-measurement states remain mutually orthogonal.

\textbf{Definition 4.} \cite{halder2019strong}
A set of orthogonal quantum states is locally irreducible if it is not possible to eliminate one or more quantum states from the set by nontrivial OPLMs.

Therefore a locally irreducible set is necessarily locally indistinguishable, although the converse does not hold.

\textbf{Definition 5.}
A set of multipartite orthogonal states is genuinely nonlocal if it is locally indistinguishable in every bipartition of the subsystems \cite{rout2019genuinely}. Furthermore, if a set is locally reducible when all parties are separated, it possesses genuine nonlocality of type I; otherwise, it exhibits type II genuine nonlocality \cite{rout2019genuinely}.

For simplicity, we omit normalization. To facilitate the discussion of local discrimination protocols, we introduce the following notation for resource configurations:
\[
\left\{\left(p,\left|\phi^+(d)\right\rangle_{AB}\right);\left(q,\left|G(d)\right\rangle_{ABC}\right)\right\},
\]
where $p$ and $q$ are non-negative real numbers denoting the average number of maximally entangled states shared between the corresponding parties, respectively \cite{rout2019genuinely}. Here,
\[
|\phi^+(d)\rangle_{AB} = \frac{1}{\sqrt{d}} \sum_{i=0}^{d-1} |ii\rangle_{AB},
\]
denotes $p$ copies of maximally entangled states consumed between $A$ and $B$. Similarly,
\[
|G(d)\rangle_{ABC} = \frac{1}{\sqrt{d}} \sum_{i=0}^{d-1} |iii\rangle_{ABC},
\]
represents $q$ copies of resources consumed among $A, B,$ and $C$. In the following, $A, B,$ and $C$ denote the three subsystems of a tripartite system corresponding to Alice, Bob, and Charlie, respectively. Correspondingly, $A_1, A_2, \ldots, A_n$ denote the $n$ subsystems of an $n$-partite system.

In entanglement-assisted discrimination protocols, we denote
\begin{equation*}
	M_i=P\bigl[(|j_{i1}\rangle,|j_{i2}\rangle)_X;|k_i\rangle_x\bigr]+P\bigl[|j_i'\rangle_X;|k_i'\rangle_x\bigr],
\end{equation*}
as a local projective measurement operator on different subsystems, where $\sum_i M_i=I$. Here, $\{|j_{i1}\rangle,|j_{i2}\rangle\}$ and $\{|j_i'\rangle\}$ are subsets of the basis of subsystem $X$, while $\{|k_i\rangle\}$ and $\{|k_i'\rangle\}$ are subsets of the basis of the auxiliary system $x$. The term
\[
P\bigl[(|j_{i1}\rangle,|j_{i2}\rangle)_X;|k_i\rangle_x\bigr]
\]
stands for
\[
\left(|j_{i1}\rangle_X\langle j_{i1}|+|j_{i2}\rangle_X\langle j_{i2}|\right)\otimes |k_i\rangle_x\langle k_i|,
\]
and $P\bigl[|j_i'\rangle_X;|k_i'\rangle_x\bigr]$ is defined analogously.

\section{ DISCRIMINATION OF THE NONLOCAL SETS IN $\mathbb{C}^{d_1} \otimes \mathbb{C}^{d_2}\left ( 3\le d_1\le d_2 \right )$}\label{Q3}
	
	In $\mathbb{C}^{3} \otimes \mathbb{C}^{5}$, the nonlocal set consisting of the following nine orthogonal product states in Ref.~\cite{Lu2024GenuinelyNS} is
	\begin{equation}\label{N1}
		\begin{aligned}
			&\left |\phi_1  \right \rangle =|1\rangle_{A}|0-1\rangle_{B},\\
			&\left |\phi_2  \right \rangle =|2\rangle_{A}|0-2\rangle_{B},\\
			&\left |\phi_3  \right \rangle =|0-1\rangle_{A}|2\rangle_{B},\\
			&\left |\phi_4  \right \rangle =|0-2\rangle_{A}|1\rangle_{B},\\
			&\left |\phi_5  \right \rangle =|0-1\rangle_{A}|3\rangle_{B},\\
			&\left |\phi_6  \right \rangle =|0-1\rangle_{A}|4\rangle_{B},\\
			&\left |\phi_7  \right \rangle =|0+1\rangle_{A}(\left |2 \right \rangle +\omega_{3} \left | 3 \right \rangle +\omega_{3}^{2}  \left | 4  \right \rangle )_{B},\\
			&\left |\phi_8  \right \rangle =|0+1\rangle_{A}(\left |2 \right \rangle +\omega_{3}^{2} \left | 3 \right \rangle +\omega_{3}  \left | 4  \right \rangle )_{B},\\
			&\left |S  \right \rangle =|0+1+2\rangle_{A}|0+1+2+3+4\rangle_{B}.\\
		\end{aligned}
	\end{equation}
	\begin{theorem}\label{thm:n1}
		The nonlocal basis given by Eq.~(\ref{N1}) can be perfectly distinguished by using the entanglement resource configuration $(1, |\phi^{+}(2)\rangle_{AB})$.
	\end{theorem}
	
	\textbf{Proof.}	For the local discrimination of these states, Alice and Bob additionally share a maximally entangled state \(\vert \phi^{+}(2)\rangle_{ab}\). The initial state is
	\[
	\begin{aligned}
		\left|\phi\right\rangle_{AB} \otimes \left|\phi^{+}(2)\right\rangle_{ab},
	\end{aligned}
	\]
	where $a$ and $b$ are the ancillary systems of Alice and Bob, respectively.
	
	\textit{Step 1.} Alice performs the measurement
	\[
	\begin{aligned}
		\mathcal{M}_1 \equiv \bigl\{
		M_{1,1}=P\bigl[(|0\rangle,|1\rangle)_A;|0\rangle_a\bigr]
		+P\bigl[|2\rangle_A;|1\rangle_a\bigr],\;
		M_{1,2}=I-M_{1,1}\bigr\}.
	\end{aligned}
	\]
	Suppose that the outcome corresponding to \(M_{1,1}\) occurs. Then
	\[
	\begin{aligned}
		&\left |\phi_1  \right \rangle \to |1\rangle_{A}|0-1\rangle_{B}|00\rangle_{ab},\\
		&\left |\phi_2  \right \rangle \to |2\rangle_{A}|0-2\rangle_{B}|11\rangle_{ab},\\
		&\left |\phi_3  \right \rangle \to |0-1\rangle_{A}|2\rangle_{B}|00\rangle_{ab},\\
		&\left |\phi_4  \right \rangle \to |0\rangle_{A}|1\rangle_{B}|00\rangle_{ab}-|2\rangle_{A}|1\rangle_{B}|11\rangle_{ab},\\
		&\left |\phi_5  \right \rangle \to |0-1\rangle_{A}|3\rangle_{B}|00\rangle_{ab},\\
		&\left |\phi_6  \right \rangle \to |0-1\rangle_{A}|4\rangle_{B}|00\rangle_{ab},\\
		&\left |\phi_7  \right \rangle \to |0+1\rangle_{A}(\left |2 \right \rangle +\omega_{3} \left | 3 \right \rangle +\omega_{3}^{2}  \left | 4  \right \rangle )_{B}|00\rangle_{ab},\\
		&\left |\phi_8  \right \rangle \to |0+1\rangle_{A}(\left |2 \right \rangle +\omega_{3}^{2} \left | 3 \right \rangle +\omega_{3}  \left | 4  \right \rangle )_{B}|00\rangle_{ab},\\
		&\left |S  \right \rangle \to |0+1\rangle_{A}|0+1+2+3+4\rangle_{B}|00\rangle_{ab}+|2\rangle_{A}|0+1+2+3+4\rangle_{B}|11\rangle_{ab}.\\
	\end{aligned}
	\]
	
	\textit{Step 2.} Bob performs the measurement
	\[
	\begin{aligned}
		\mathcal{M}_2\equiv \bigl\{
		M_{2,1}=P[(|2\rangle,|3\rangle,|4\rangle)_{B};|0\rangle_{b}],
		\;
		M_{2,2}=I-M_{2,1}
		\bigr\}.
	\end{aligned}
	\]
	If outcome \(M_{2,1}\) occurs, then the states are \(|\phi_{3} \rangle,~|\phi_{5} \rangle,~|\phi_{6} \rangle,~|\phi_{7} \rangle,~|\phi_{8} \rangle,\) and $\left |S  \right \rangle \to|0+1\rangle_{A}|2+3+4\rangle_{B}|00\rangle_{ab}$. Otherwise, the state is one of the remaining four states \(|\phi_{1} \rangle,~|\phi_{2} \rangle,~|\phi_{4} \rangle,\) and $\left |S  \right \rangle \to|0+1\rangle_{A}|0+1\rangle_{B}|00\rangle_{ab}+|2\rangle_{A}|0+1+2+3+4\rangle_{B}|11\rangle_{ab}$.
	
	\textit{Step 3.} For the branch corresponding to  \(M_{2,1}\), Alice performs the measurement
	\[
	\begin{aligned}
		\mathcal{M}_3\equiv \bigl\{
		M_{3,1}=P[|0-1\rangle_A;|0\rangle_a],
		\;
		M_{3,2}=I-M_{3,1}
		\bigr\}.
	\end{aligned}
	\]
	If outcome \(M_{3,1}\) occurs, then the state belongs to set \(\{|\phi_{3} \rangle,|\phi_{5} \rangle,|\phi_{6} \rangle\}\), which can be perfectly distinguished by measuring subsystem B; if outcome \(M_{3,2}\) occurs, then the state belongs to set \(\{|\phi_{7} \rangle,|\phi_{8} \rangle, \left|S  \right \rangle \to|0+1\rangle_{A}|2+3+4\rangle_{B}|00\rangle_{ab}\}\) which is locally distinguishable.
	
	For the branch corresponding to \(M_{2,2}\), Alice then performs the measurement
	\[
	\begin{aligned}
		\mathcal{M}_4\equiv \bigl\{
		M_{4,1}=P[|1\rangle_A;|0\rangle_a],
		\;
		M_{4,2}=I-M_{4,1}
		\bigr\}.
	\end{aligned}
	\]
	If outcome \(M_{4,1}\) occurs, the states are \(|\phi_{1} \rangle\) and $\left |S \right  \rangle \to|1\rangle_{A}|0+1\rangle_{B}|00\rangle_{ab}$, which are LOCC distinguishable. For outcome \(M_{4,2}\), the remaining possibilities are \(|\phi_{2} \rangle,~|\phi_{4} \rangle,\) and $\left |S \right  \rangle \to|0\rangle_{A}|0+1\rangle_{B}|00\rangle_{ab}+|2\rangle_{A}|0+1+2+3+4\rangle_{B}|11\rangle_{ab}$. Next, Bob performs
	\[
	\mathcal{M}_4' \equiv \bigl\{
	M_{4,1}'=P[|1\rangle_{B};I_{b}],
	\;
	M_{4,2}'=I-M_{4,1}'
	\bigr\}.
	\]
	The corresponding results for \(M_{4,1}'\) and \(M_{4,2}'\) are \(\{|\phi_{4} \rangle, \left |S \right   \rangle\to |0\rangle_{A}|1\rangle_{B}|00\rangle_{ab}+|2\rangle_{A}|1\rangle_{B}|11\rangle_{ab}\}\) and \(\{|\phi_{2} \rangle, \left |S \right \rangle\to |0\rangle_{A}|0\rangle_{B}|00\rangle_{ab}+|2\rangle_{A}|0+2+3+4\rangle_{B}|11\rangle_{ab}\}\), respectively. These sets are all locally distinguishable.
	
	If other outcome occurs in Step 1, one can construct a similar protocol to locally distinguish the states in set (\ref{N1}).
	\qed
	
	In this protocol, we consume a total of 1 ebit of entanglement. We also find that the set of the obtained states and the corresponding $|S\rangle$ can always be locally discriminated after performing the orthogonality-preserving projection measurement. Next, we will generalize these results to $d_1 \otimes d_2 $ systems and no longer discuss the state $|S\rangle$.
	
	In the system $\mathbb{C}^{d_1} \otimes \mathbb{C}^{d_2}\left (3 \le d_2 \le d_1\right )$, the nonlocal set of orthogonal product states in Ref.~\cite{Lu2024GenuinelyNS} is
	\begin{equation}\label{N2}
		\begin{aligned}
			&|\phi_i\rangle = |i\rangle_A |0-i\rangle_B, \quad 1 \leq i \leq d_1-1, \\
			&|\phi_{i+(d_1-1)}\rangle = |0-i\rangle_A |m\rangle_B, \quad 1 \leq i \leq d_1-2, \ m = i+1; \ i = d_1-1, \ m = 1, \\
			&|\phi_{j+(d_1-1)}\rangle = |0-1\rangle_A |j\rangle_B, \quad d_1 \leq j \leq d_2-1, \\
			&|\phi_{s+(d_1+d_2-2)}\rangle = |0+1\rangle_A \Bigl(|2\rangle_B + \sum_{t=1}^{d_2-d_1} \omega_{d_2-d_1+1}^{st} |t+d_1-1\rangle_B \Bigr),  \quad 1 \leq s \leq d_2-d_1, \\
			&|S\rangle = |0+1+\cdots+(d_1-1)\rangle_A |0+1+\cdots+(d_2-1)\rangle_B.
		\end{aligned}
	\end{equation}
	
	\begin{theorem}\label{thm:n2}
		The entanglement resource configuration $(1,\vert \phi^{+}(3)\rangle_{AB})$ is sufficient for local discrimination of the nonlocal set~(\ref{N2}).
	\end{theorem}
	
	The detailed procedure is given in Appendix~\ref{A}. Since set (\ref{N1}) and set (\ref{N2}) share the same structural features, the distinguishing scheme in Theorem \ref{thm:n2} is similar to that in Theorem \ref{thm:n1}. The difference is that set (\ref{N2}) contains extra quantum states $\{|\phi_i\rangle\}_{i=3}^{d_1-1}$ and $\{|\phi_{i+(d_1-1)}\rangle\}_{i=2}^{d_1-2}$ compared with set (\ref{N1}) when $d_1$ and $d_2$ are relatively large. EPR state alone is insufficient to accomplish the local discrimination of (\ref{N2}), while the maximally entangled state $|\phi^{+}(3)\rangle$ can realize the task. Therefore, at most $\log_{2}3$ ebits of entanglement is consumed in Theorem \ref{thm:n2}. In fact, a single EPR state suffices when $d_1$ and $d_2$ are relatively small.
	
	\section{ DISCRIMINATION OF THE ORTHOGONAL PRODUCT SETS WITH GENUINE NONLOCALITY OF TYPE I} \label{Q4}
	\subsection{ Genuinely nonlocal sets in $\mathbb{C}^{d_1} \otimes \mathbb{C}^{d_2} \otimes \mathbb{C}^{d_2}\left ( 3\le (d_1-1)\le d_2\le d_3 \right )$}
	Building upon the nonlocal set in $\mathbb{C}^{d_1}\otimes\mathbb{C}^{d_2}$ introduced above, Lu et al.~\cite{Lu2024GenuinelyNS} further constructed a genuinely nonlocal set of type I in the tripartite system $\mathbb{C}^{d_1}\otimes\mathbb{C}^{d_2}\otimes\mathbb{C}^{d_3}$, where $3 \leq d_1 -1 \leq d_2 \leq d_3$. In the system $\mathbb{C}^{4} \otimes \mathbb{C}^{4} \otimes \mathbb{C}^{6}$, the genuinely nonlocal set consisting of the following 18 orthogonal product states is
	\begin{equation}\label{GNI1}
		\begin{aligned}
			&|\phi_1\rangle = |1\rangle_A |0-1\rangle_B |0\rangle_C, \\
			&|\phi_2\rangle = |2\rangle_A |0-2\rangle_B |0\rangle_C, \\
			&|\phi_3\rangle = |0-1\rangle_A |2\rangle_B |0\rangle_C, \\
			&|\phi_4\rangle = |0-2\rangle_A |1\rangle_B |0\rangle_C, \\
			&|\phi_5\rangle = |0-1\rangle_A |3\rangle_B |0\rangle_C, \\
			&|\phi_6\rangle = |0+1\rangle_A |2-3\rangle_B |0\rangle_C, \\
			&|\phi_7\rangle = |0+1+2\rangle_A |0+1+2+3\rangle_B |0\rangle_C, \\
			&|\phi_8\rangle = |3\rangle_A |1\rangle_B |0-1\rangle_C, \\
			&|\phi_9\rangle = |3\rangle_A |2\rangle_B |0-2\rangle_C, \\
			&|\phi_{10}\rangle = |3\rangle_A |3\rangle_B |0-3\rangle_C, \\
			&|\phi_{11}\rangle = |3\rangle_A |0-1\rangle_B |2\rangle_C, \\
			&|\phi_{12}\rangle = |3\rangle_A |0-2\rangle_B |3\rangle_C, \\
			&|\phi_{13}\rangle = |3\rangle_A |0-3\rangle_B |1\rangle_C, \\
			&|\phi_{14}\rangle = |3\rangle_A |0-1\rangle_B |4\rangle_C, \\
			&|\phi_{15}\rangle = |3\rangle_A |0-1\rangle_B |5\rangle_C, \\
			&|\phi_{16}\rangle = |3\rangle_A |0+1\rangle_B \bigl(|2\rangle_C + \omega_3 |4\rangle_C + \omega_3^2 |5\rangle_C\bigr), \\
			&|\phi_{17}\rangle = |3\rangle_A |0+1\rangle_B \bigl(|2\rangle_C + \omega_3^2 |4\rangle_C + \omega_3 |5\rangle_C\bigr), \\
			&|\phi_{18}\rangle = |3\rangle_A |0+1+2+3\rangle_B |0+1+\cdots+5\rangle_C.
		\end{aligned}
	\end{equation}
	
	\begin{theorem}\label{thm:GNI1}
		The genuinely nonlocal set  ${\textstyle \bigcup_{i=1}^{18}} \left | \phi_{i}   \right \rangle $ in Eq. (\ref{GNI1}) can be locally distinguished by using the entanglement resource
		configuration \(\{(\frac{7}{18},\vert \phi^{+}(2)\rangle_{AB});(\frac{11}{18} ,\vert \phi^{+}(2)\rangle_{BC})\}\).
	\end{theorem}
	
	\textbf{Proof.} Since this set is locally reducible, when Alice performs the measurements\[
	\begin{aligned}
		\mathcal{M}_1 \equiv \bigl\{
		M_{1,1}=|3\rangle_A \langle 3|,\;
		M_{1,2}=I-M_{1,1}\bigr\},
	\end{aligned}
	\]the whole set can be locally reduced to two disjoint subsets. The corresponding results for \(M_{1,1}\) and \(M_{1,2}\) are \(\left \{ \left | \phi _{i}   \right \rangle  \right \} _{i=8}^{18} \) and \(\left \{ \left | \phi _{i}   \right \rangle  \right \} _{i=1}^{7} \), respectively.
	We then proceed as follows.
	
	\textit{Step 1.} For the branch corresponding to \(M_{1,2}\), to locally distinguish the states, Alice and Bob share a maximally entangled state \(\vert \phi^{+}(2)\rangle_{ab}\). The initial state is
	\[
	\begin{aligned}
		\left|\phi\right\rangle_{ABC} \otimes \left|\phi^{+}(2)\right\rangle_{ab_1},
	\end{aligned}
	\] where $a$ and $b_1$ are the ancillary systems of Alice and Bob, respectively. Then Alice performs the measurement
	\[
	\begin{aligned}
		\mathcal{M}_2 \equiv \bigl\{
		M_{2,1}=P\bigl[(|0\rangle,|1\rangle)_A;|0\rangle_a\bigr]+P\bigl[|2\rangle_A;|1\rangle_a\bigr],\;
		M_{2,2}=I-M_{2,1}\bigr\}.
	\end{aligned}
	\]
	Suppose that the outcome corresponding to \(M_{2,1}\) occurs. Then
	\[
	\begin{aligned}
		&|\phi_1\rangle \to  |1\rangle_A |0-1\rangle_B |0\rangle_C|00\rangle_{ab_1}, \\
		&|\phi_2\rangle \to  |2\rangle_A |0-2\rangle_B |0\rangle_C|11\rangle_{ab_1}, \\
		&|\phi_3\rangle \to  |0-1\rangle_A |2\rangle_B |0\rangle_C|00\rangle_{ab_1}, \\
		&|\phi_4\rangle \to  |0\rangle_A |1\rangle_B |0\rangle_C|00\rangle_{ab_1}-|2\rangle_A |1\rangle_B |0\rangle_C|11\rangle_{ab_1}, \\
		&|\phi_5\rangle \to  |0-1\rangle_A |3\rangle_B |0\rangle_C|00\rangle_{ab_1}, \\
		&|\phi_6\rangle \to  |0+1\rangle_A |2-3\rangle_B |0\rangle_C|00\rangle_{ab_1}, \\
		&|\phi_7\rangle \to  |0+1\rangle_A |0+1+2+3\rangle_B |0\rangle_C|00\rangle_{ab_1}+|2\rangle_A |0+1+2+3\rangle_B |0\rangle_C|11\rangle_{ab_1}. \\
	\end{aligned}
	\]
	
	\textit{Step 2.} Bob performs the measurement
	\[
	\begin{aligned}
		\mathcal{M}_3\equiv \bigl\{
		M_{3,1}=P[(|2\rangle,|3\rangle)_{B};|0\rangle_{b_1}],
		\;
		M_{3,2}=I-M_{3,1}
		\bigr\}.
	\end{aligned}
	\]
	If outcome \(M_{3,1}\) occurs, then the states are \(|\phi_{3}\rangle ,~|\phi_{5} \rangle,~|\phi_{6} \rangle,\) and $\left |\phi_{7} \right \rangle \to|0+1\rangle_{A}|2+3\rangle_{B}|0\rangle_C|00\rangle_{ab_1}$, which are locally distinguishable. Otherwise, the state belongs to the set of remaining states \(\{|\phi_{1}\rangle ,|\phi_{2}\rangle ,|\phi_{4} \rangle,\left |\phi_{7}  \right \rangle \to|0+1\rangle_{A}|0+1\rangle_{B}|0\rangle_C|00\rangle_{ab_1}+|2\rangle_{A}|0+1+2+3\rangle_{B}|0\rangle_C|11\rangle_{ab_1}\).
	
	\textit{Step 3.} Alice performs the measurement
	\[
	\begin{aligned}
		\mathcal{M}_4\equiv \bigl\{
		M_{4,1}=P[|1\rangle_A;|0\rangle_a],
		\;
		M_{4,2}=I-M_{4,1}
		\bigr\}.
	\end{aligned}
	\]
	If outcome \(M_{4,1}\) occurs, then the states  are \(|\phi_{1} \rangle\) and $\left |\phi_{7}  \right \rangle \to|1\rangle_{A}|0+1\rangle_{B}|0\rangle_C|00\rangle_{ab_1}$, which are LOCC distinguishable; if outcome \(M_{4,2}\) occurs, then the states are \(|\phi_{2}\rangle ,~|\phi_{4}\rangle,\) and $\left |\phi_{7}  \right \rangle \to|0\rangle_{A}|0+1\rangle_{B}|0\rangle_C|00\rangle_{ab_1}+|2\rangle_{A}|0+1+2+3\rangle_{B}|0\rangle_C|11\rangle_{ab_1}$. Next, Bob performs
	\[
	\mathcal{M}_4' \equiv \bigl\{
	M_{4,1}'=P[|1\rangle_{B};I_{b_1}],
	\;
	M_{4,2}'=I-M_{4,1}'
	\bigr\}.
	\]
	The corresponding results for \(M_{4,1}'\) and \(M_{4,2}'\) are \(\{|\phi_{4} \rangle, \left |S_{7} \right   \rangle \to|0\rangle_{A}|1\rangle_{B}|0\rangle_C|00\rangle_{ab_1}+|2\rangle_{A}|1\rangle_{B}|0\rangle_C|11\rangle_{ab_1}\}\) and \(\{|\phi_{2} \rangle,\left |S_{7}  \right \rangle \to|0\rangle_{A}|0\rangle_{B}|0\rangle_C|00\rangle_{ab_1}+|2\rangle_{A}|0+2+3+4\rangle_{B}|0\rangle_C|11\rangle_{ab_1}\}\), respectively. These sets are all LOCC distinguishable.
	
	\textit{Step 1$'$.} For the branch corresponding to \(M_{1,1}\), Alice and Bob share a maximally entangled state \(\vert \phi^{+}(2)\rangle_{b_2c}\). The initial state is
	\[
	\begin{aligned}
		\left|\phi\right\rangle_{ABC} \otimes \left|\phi^{+}(2)\right\rangle_{b_2c},
	\end{aligned}
	\] where $b_2$ and $c$ are the ancillary systems of Bob and Charlie, respectively. Then Bob performs the measurement
	\[
	\begin{aligned}
		\mathcal{M}_5 \equiv \bigl\{
		M_{5,1}=P\bigl[(|0\rangle,|1\rangle)_B;|0\rangle_{b_2}\bigr]+P\bigl[(|2\rangle,|3\rangle)_B;|1\rangle_{b_2}\bigr],\;
		M_{5,2}=I-M_{5,1}\bigr\}.
	\end{aligned}
	\]
	Suppose that the outcome corresponding to \(M_{5,1}\) occurs. Then
	\[
	\begin{aligned}
		&|\phi_8\rangle \to  |3\rangle_A |1\rangle_B |0-1\rangle_C|00\rangle_{b_2c}, \\
		&|\phi_9\rangle \to  |3\rangle_A |2\rangle_B |0-2\rangle_C|11\rangle_{b_2c}, \\
		&|\phi_{10}\rangle \to  |3\rangle_A |3\rangle_B |0-3\rangle_C|11\rangle_{b_2c}, \\
		&|\phi_{11}\rangle \to  |3\rangle_A |0-1\rangle_B |2\rangle_C|00\rangle_{b_2c}, \\
		&|\phi_{12}\rangle \to  |3\rangle_A |0\rangle_B |3\rangle_C|00\rangle_{b_2c}-|3\rangle_A |2\rangle_B |3\rangle_C|11\rangle_{b_2c}, \\
		&|\phi_{13}\rangle \to  |3\rangle_A |0\rangle_B |1\rangle_C|00\rangle_{b_2c}-|3\rangle_A |3\rangle_B |1\rangle_C|11\rangle_{b_2c}, \\
		&|\phi_{14}\rangle \to  |3\rangle_A |0-1\rangle_B |4\rangle_C|00\rangle_{b_2c}, \\
		&|\phi_{15}\rangle \to  |3\rangle_A |0-1\rangle_B |5\rangle_C|00\rangle_{b_2c}, \\
		&|\phi_{16}\rangle \to  |3\rangle_A |0+1\rangle_B \bigl(|2\rangle_C + \omega_3 |4\rangle_C + \omega_3^2 |5\rangle_C\bigr)|00\rangle_{b_2c}, \\
		&|\phi_{17}\rangle \to |3\rangle_A |0+1\rangle_B \bigl(|2\rangle_C + \omega_3^2 |4\rangle_C + \omega_3 |5\rangle_C\bigr)|00\rangle_{b_2c}, \\
		&|\phi_{18}\rangle \to |3\rangle_A |0+1\rangle_B |0+1+\cdots+5\rangle_C|00\rangle_{b_2c}+|3\rangle_A |2+3\rangle_B |0+1+\cdots+5\rangle_C|11\rangle_{b_2c}.
	\end{aligned}
	\]
	
	\textit{Step 2$'$.} Charlie then performs the measurement
	\[
	\begin{aligned}
		\mathcal{M}_6\equiv \bigl\{
		M_{6,1}=P[(|2\rangle,|4\rangle,|5\rangle)_C;|0\rangle_c],
		\;
		M_{6,2}=I-M_{6,1}
		\bigr\}.
	\end{aligned}
	\]
	If outcome \(M_{6,1}\) occurs, the states are \(|\phi_{11}\rangle ,~|\phi_{14}\rangle ,~|\phi_{15} \rangle,~|\phi_{16}\rangle ,~|\phi_{17} \rangle,\) and $\left |\phi_{18} \right  \rangle \to|3\rangle_A |0+1\rangle_B |2+4+5\rangle_C|00\rangle_{b_2c}$. For outcome \(M_{6,2}\), the remaining possibilities are \(|\phi_{8}\rangle ,~|\phi_{9}\rangle ,~|\phi_{10} \rangle,~|\phi_{12}\rangle ,~|\phi_{13} \rangle,\) and $\left |\phi_{18} \right  \rangle \to|3\rangle_A |0+1\rangle_B |0+1+3\rangle_C|00\rangle_{b_2c}+|3\rangle_A |2+3\rangle_B |0+1+\cdots+5\rangle_C|11\rangle_{b_2c}$.
	
	\textit{Step 3$'$.} For the branch corresponding to \(M_{6,1}\), Bob performs
	\[
	\mathcal{M}_7 \equiv \bigl\{
	M_{7,1}=P[|0-1\rangle_{B};|0\rangle_{b_2}],
	\;
	M_{7,2}=I-M_{7,1}
	\bigr\}.
	\]
	The corresponding results for \(M_{7,1}\) and \(M_{7,2}\) are \(\{|\phi_{11}\rangle ,|\phi_{14}\rangle ,|\phi_{15} \rangle\}\) and \(\{|\phi_{16}\rangle ,|\phi_{17} \rangle,|\phi_{18} \rangle\}\), respectively. They are all locally distinguishable.
	
	For the branch corresponding to \(M_{6,2}\), Bob performs
	\[
	\mathcal{M}_8 \equiv \bigl\{
	M_{8,1}=P[|1\rangle_{B};|0\rangle_{b_2}],
	\;
	M_{8,2}=I-M_{8,1}
	\bigr\}.
	\]
	The corresponding results for \(M_{8,1}\) are \(|\phi_{8} \rangle\) and $\left |\phi_{18}\right  \rangle \to|3\rangle_A |1\rangle_B |0+1+3\rangle_C|00\rangle_{b_2c}$, which are LOCC distinguishable. Next, Charlie performs
	\[
	\mathcal{M}_9 \equiv \bigl\{
	M_{9,1}=P[|1\rangle_{C};I_{c}],
	\;
	M_{9,2}=I-M_{9,1}
	\bigr\}.
	\]
	If outcome \(M_{9,1}\) occurs, the states are \(|\phi_{13} \rangle\) and $\left |\phi_{18} \right  \rangle \to|3\rangle_A |0\rangle_B |1\rangle_C|00\rangle_{b_2c}+|3\rangle_A |2+3\rangle_B |1\rangle_C|11\rangle_{b_2c}$. For outcome \(M_{9,2}\), the remaining possibilities are \(|\phi_{9} \rangle,~|\phi_{10} \rangle,~|\phi_{12} \rangle,\) and $\left |\phi_{18} \right  \rangle \to|3\rangle_A |0\rangle_B |0+3\rangle_C|00\rangle_{b_2c}+|3\rangle_A |2+3\rangle_B |0+2+\cdots+5\rangle_C|11\rangle_{b_2c}$. Then, Bob performs
	\[
	\mathcal{M}_{10} \equiv \bigl\{
	M_{10,1}=P[|3\rangle_{B};|1\rangle_{b_2}],
	\;
	M_{10,2}=I-M_{10,1}
	\bigr\}.
	\]
	The corresponding results for \(M_{10,1}\) are \(|\phi_{10} \rangle\) and $\left |\phi_{18}\right  \rangle \to|3\rangle_A |3\rangle_B |0+2+3+4+5\rangle_C|11\rangle_{b_2c}$, which are LOCC distinguishable. Otherwise, the state is one of the remaining possibilities \(|\phi_{9} \rangle\) and \(|\phi_{12} \rangle\). They are all locally distinguishable.
	
	If outcome $M_{2,2}$ in Step 1 or outcome $M_{5,2}$ in Step 1$'$ occurs, a similar method gives an LOCC protocol for distinguishing the states in set (\ref{GNI1}).
	\qed
	
	Since the set is locally reducible, we first divide it into two subsets. We then discriminate each subset separately. Noting that the two subsets possess structure analogous to the nonlocal set in Eq.~(\ref{N2}). So, we establish entanglement resources between the corresponding subsystems. In this protocol, we consume a total of 1 ebit of entanglement. We then extend the protocol to the \( \mathbb{C}^{d_1} \otimes \mathbb{C}^{d_2}\otimes \mathbb{C}^{d_3} \) system.
	
	In $\mathbb{C}^{d_1} \otimes \mathbb{C}^{d_2}\otimes \mathbb{C}^{d_3}\left ( 3\le (d_1-1)\le d_2\le d_3 \right )$, the genuinely nonlocal set
	of type I in Ref.~\cite{Lu2024GenuinelyNS} is
	\begin{equation}\label{GNI2}
		\begin{aligned}
			&|\phi_i\rangle = |i\rangle_A |0-i\rangle_B |0\rangle_C, \quad 1 \leq i \leq d_1-2, \\
			&|\phi_{i+(d_1-2)}\rangle = |0-i\rangle_A |m\rangle_B |0\rangle_C, \quad 1 \leq i \leq d_1-3, \ m = i+1; \ i = d_1-2, \ m = 1, \\
			&|\phi_{j+(d_1-2)}\rangle = |0-1\rangle_A |j\rangle_B |0\rangle_C, \quad d_1-1 \leq j \leq d_2-1, \\
			&|\phi_{s+(d_1+d_2-3)}\rangle = |0+1\rangle_A |\eta\rangle_B |0\rangle_C, \quad 1 \leq s \leq d_2-d_1+1, \\
			&|\phi_{2d_2-1}\rangle = |0+1+\cdots+(d_1-2)\rangle_A  |0+1+\cdots+(d_2-1)\rangle_B |0\rangle_C,\\
			&|\phi_{i+2d_2-1}\rangle = |d_1-1\rangle_A |i\rangle_B |0-i\rangle_C, \quad 1 \leq i \leq d_2-1, \\
			&|\phi_{i+3d_2-2}\rangle = |d_1-1\rangle_A |0-i\rangle_B |m\rangle_C, \quad 1 \leq i \leq d_2-2,  \ m = i+1; \ i = d_2-1, \ m = 1, \\
			&|\phi_{j+3d_2-2}\rangle = |d_1-1\rangle_A |0-1\rangle_B |j\rangle_C, \quad d_2 \leq j \leq d_3-1, \\
			&|\phi_{s+3d_2+d_3-3}\rangle = |d_1-1\rangle_A |0+1\rangle_B |\eta\rangle_C, \quad 1 \leq s \leq d_3-d_2, \\
			&|\phi_{2(d_2+d_3)-2}\rangle = |d_1-1\rangle_A |0+1+\cdots+(d_2-1)\rangle_B |0+1+\cdots+(d_3-1)\rangle_C,
		\end{aligned}
	\end{equation}
	where $|\eta\rangle_B = |2\rangle_B + \sum_{t=1}^{d_2-d_1+1} \omega_{d_2-d_1+2}^{st} |t+d_1-2\rangle_B$, $|\eta\rangle_C = |2\rangle_C + \sum_{t=1}^{d_3-d_2} \omega_{d_3-d_2+1}^{st} |t+d_2-1\rangle_C$.
	
	\begin{theorem}\label{thm:GNI2}
		The set in Eq.~(\ref{GNI2}) can be locally distinguished by using the resource configuration \(\{(\frac{r}{s},\vert \phi^{+}(3)\rangle_{AB});(\frac{r'}{s},\vert \phi^{+}(3)\rangle_{BC})\}\),  where \[
		\begin{aligned}
			s= 2\left (d_2+d_3 \right )-2 ,\;
			r=2 d_2-1   ,\;
			r'=2d_3-1  .
		\end{aligned}
		\]
	\end{theorem}
	
	We present the detailed procedure in Appendix~\ref{B}. Similarly, we first divide the set into two parts \(\left \{ \left | \phi _{i}   \right \rangle  \right \} _{i=1}^{2d_2-1}\)  and \(\left \{ \left | \phi _{i}   \right \rangle  \right \} _{i=2d_2}^{  2\left (d_2+d_3 \right )-2 }\). Then, referring to the discrimination protocol on the $\mathbb{C}^{d_1} \otimes \mathbb{C}^{d_2}$ system in Theorem \ref{thm:n2}, we let different subsystems share the $|\phi^+(3)\rangle$ state in each part. In this protocol, we consume a total of $\log_{2}{3} $ ebits of entanglement. Next, we consider the orthogonal product sets with genuine nonlocality of type I in \( \bigotimes_{i=1}^{n} \mathbb{C}^{d_i} \).
	\subsection{Genuinely nonlocal sets in $\otimes_{i=1}^{n} \mathbb{C}^{d_i}$ $(3 \leq (d_1-1) \leq d_2 \leq \dots \leq d_n, n \geq 3)$}
	In $\otimes_{i=1}^{n} \mathbb{C}^{d_i}$ $(3 \leq (d_1-1) \leq d_2 \leq \dots \leq d_n, n \geq 3)$,
	the genuinely nonlocal set of type I~\cite{Lu2024GenuinelyNS} is
	\begin{equation}\label{GNI3}
		\begin{aligned}
			&G_1 = \bigl\{|x_{m_1}\rangle_{A_1}|y_{m_1}\rangle_{A_2}|0\rangle_{A_3}|0\rangle_{A_4}\cdots |0\rangle_{A_{n-3}}|0\rangle_{A_{n-2}}|0\rangle_{A_{n-1}}|1\rangle_{A_n}\bigr\}, \\
			&G_2 = \bigl\{|1\rangle_{A_1}|x_{j_2}\rangle_{A_2}|z_{j_2}\rangle_{A_3}|0\rangle_{A_4}\cdots |0\rangle_{A_{n-3}}|0\rangle_{A_{n-2}}|0\rangle_{A_{n-1}}|0\rangle_{A_n}\bigr\}, \\
			&G_3 = \bigl\{|0\rangle_{A_1}|x_{j_3}\rangle_{A_2}|1\rangle_{A_3}|z_{j_3}\rangle_{A_4}\cdots |0\rangle_{A_{n-3}}|0\rangle_{A_{n-2}}|0\rangle_{A_{n-1}}|0\rangle_{A_n}\bigr\}, \\
			&~~\vdots \\
			&G_{n-2} = \bigl\{|0\rangle_{A_1}|x_{j_{n-2}}\rangle_{A_2}|0\rangle_{A_3}|0\rangle_{A_4}\cdots |0\rangle_{A_{n-3}}|1\rangle_{A_{n-2}}|z_{j_{n-2}}\rangle_{A_{n-1}}|0\rangle_{A_n}\bigr\}, \\
			&G_{n-1} = \bigl\{|d_1-1\rangle_{A_1}|x_{j_{n-1}}\rangle_{A_2}|0\rangle_{A_3}|0\rangle_{A_4}\cdots |0\rangle_{A_{n-3}}|0\rangle_{A_{n-2}}|1\rangle_{A_{n-1}}|z_{j_{n-1}}\rangle_{A_n}\bigr\},
		\end{aligned}
	\end{equation}
	where $\{|x_{m_1}\rangle|y_{m_1}\rangle, m_1 = 1,2,\dots,2d_2-1\}$ are $2d_2-1$
	locally indistinguishable orthogonal product states in
	$\mathbb{C}^{d_1-1}\otimes\mathbb{C}^{d_2}$. And $\{|x_{j_{k-1}}\rangle|z_{j_{k-1}}\rangle,
	j_{k-1} = 1,2,\dots,2d_k-1\}$ are $2d_k-1$ locally indistinguishable
	orthogonal product states in $\mathbb{C}^{d_2}\otimes\mathbb{C}^{d_k}$
	($k = 3,\cdots,n$).
	
	\begin{theorem}\label{thm:GNI3}
		The set $\bigcup_{i=1}^{n-1}G_i$ provided by Eq. (\ref{GNI3}) can be locally distinguished by using the resource configuration \(\{(\frac{r_1}{s},\vert \phi^{+}(3)\rangle_{A_1A_2});(\frac{r_2}{s},\vert \phi^{+}(3)\rangle_{A_2A_3});(\frac{r_3}{s},\vert \phi^{+}(3)\rangle_{A_2A_4});\cdots;(\frac{r_{n-1}}{s},\vert \phi^{+}(3)\rangle_{A_2A_n})\}\),  where \[
		\begin{aligned}
			s= \sum_{k=2}^{n} \left ( 2d_k-1 \right ),~r_i=2 d_{i+1}-1~for~i=1,2,\ldots,n-1.
		\end{aligned}
		\]
	\end{theorem}
	
	Since each subset $G_i$ $(i=1,2,\ldots,n-1)$ can be perfectly distinguished via a protocol analogous to that given in Theorem \ref{N2}, the problem reduces to distinguishing among the subsets $G_i$. Moreover, we observe that the set in Eq.~(\ref{GNI3}) is still locally reducible, no additional entanglement is required at this step. After these subsets $G_i$ $(i=1,2,\ldots,n-1)$ are discriminated, we allow different subsystems to share entangled state \( \vert \phi^{+}(3) \rangle \) for distinct subsets $G_i$, whereby the discrimination task is completed. So, the protocol in Theorem \ref{thm:GNI3} consumes a total of \(\log_{2}{3} \) ebits of entanglement. The detailed procedure is given in Appendix~\ref{C}.
	
	\section{DISCRIMINATION OF THE ORTHOGONAL PRODUCT SETS WITH GENUINE NONLOCALITY OF TYPE II}\label{Q5}
	\subsection{ Genuinely nonlocal sets in $\mathbb{C}^{d_1} \otimes \mathbb{C}^{d_2} \otimes \mathbb{C}^{d_3}$ ($3 \leq d_1 \leq d_2 \leq d_3 )$}
	In this section we will increase the number of auxiliary subsystems to explore the local discrimination protocol. In Ref.~\cite{Lu2024GenuinelyNS}, genuinely nonlocal sets of type II were constructed. Here, we propose local discrimination protocols by using the states $|G(2)\rangle=\frac{1}{\sqrt{2}}(|000\rangle+|111\rangle)$ and $|G(3)\rangle=\frac{1}{\sqrt{3}}(|000\rangle+|111\rangle=|000\rangle+|222\rangle)$. In $\mathbb{C}^3 \otimes \mathbb{C}^4 \otimes \mathbb{C}^5$, the type-II genuinely nonlocal set consisting of the following 14 orthogonal product states is
	\begin{equation}\label{GNII1}
		\begin{aligned}
			&|\phi_1\rangle = |1\rangle_A |0-1\rangle_B |1\rangle_C, \\
			&|\phi_2\rangle = |2\rangle_A |0-2\rangle_B |1\rangle_C, \\
			&|\phi_3\rangle = |0-1\rangle_A |2\rangle_B |1\rangle_C, \\
			&|\phi_4\rangle = |0-2\rangle_A |1\rangle_B |1\rangle_C, \\
			&|\phi_5\rangle = |0-1\rangle_A |3\rangle_B |1\rangle_C, \\
			&|\phi_6\rangle = |0+1\rangle_A |2-3\rangle_B |1\rangle_C, \\
			&|\phi_7\rangle = |0+1+2\rangle_A |0+1+2+3\rangle_B |0+1+2+3+4\rangle_C, \\
			&|\phi_8\rangle = |1\rangle_A |0+1\rangle_B |0-1\rangle_C, \\
			&|\phi_9\rangle = |2\rangle_A |0+1\rangle_B |0-2\rangle_C, \\
			&|\phi_{10}\rangle = |0-1\rangle_A |0+1\rangle_B |2\rangle_C, \\
			&|\phi_{11}\rangle = |0-1\rangle_A |0+1\rangle_B |3\rangle_C, \\
			&|\phi_{12}\rangle = |0-1\rangle_A |0+1\rangle_B |4\rangle_C, \\
			&|\phi_{13}\rangle = |0+1\rangle_A |0+1\rangle_B \bigl(|2\rangle_C + \omega_3 |3\rangle_C + \omega_3^2 |4\rangle_C\bigr), \\
			&|\phi_{14}\rangle = |0+1\rangle_A |0+1\rangle_B \bigl(|2\rangle_C + \omega_3^2 |3\rangle_C + \omega_3 |4\rangle_C\bigr).
		\end{aligned}
	\end{equation}
	
	\begin{theorem}\label{thm:GNII1}
		The set ${\textstyle \bigcup_{i=1}^{14}} \left | \phi_{i}   \right \rangle $ in Eq.(\ref{GNII1}) can be locally distinguished by using the resource configuration $(1, |G(2)\rangle_{ABC})$.
	\end{theorem}
	
	\textbf{Proof.}	 Using the entanglement resource \(|G(2) \rangle_{abc}=|000\rangle_{abc}+|111\rangle_{abc}\), the initial state is
	\[
	\begin{aligned}
		\left|\phi\right\rangle_{ABC} \otimes \left|G(2)\right\rangle_{abc},
	\end{aligned}
	\]where $a$, $b$, and $c$ are the ancillary systems of Alice, Bob, and Charlie, respectively.
	
	\textit{Step 1.} Alice performs the measurement
	\[
	\begin{aligned}
		\mathcal{M}_1 \equiv \bigl\{
		M_{1,1}=P\bigl[(|0\rangle,|1\rangle)_A;|0\rangle_a\bigr]
		+P\bigl[|2\rangle_A;|1\rangle_a\bigr],\;
		M_{1,2}=I-M_{1,1}\bigr\}.
	\end{aligned}
	\]
	Suppose at the outcome corresponding to \(M_{1,1}\) occurs. Then
	\[
	\begin{aligned}
		&|\phi_1\rangle \to  |1\rangle_A |0-1\rangle_B |1\rangle_C|000\rangle_{abc}, \\
		&|\phi_2\rangle \to  |2\rangle_A |0-2\rangle_B |1\rangle_C|111\rangle_{abc}, \\
		&|\phi_3\rangle \to  |0-1\rangle_A |2\rangle_B |1\rangle_C|000\rangle_{abc}, \\
		&|\phi_4\rangle \to  |0\rangle_A |1\rangle_B |1\rangle_C|000\rangle_{abc}-|2\rangle_A |1\rangle_B |1\rangle_C|111\rangle_{abc}, \\
		&|\phi_5\rangle \to  |0-1\rangle_A |3\rangle_B |1\rangle_C|000\rangle_{abc}, \\
		&|\phi_6\rangle \to  |0+1\rangle_A |2-3\rangle_B |1\rangle_C|000\rangle_{abc}, \\
		&|\phi_7\rangle \to  |0+1\rangle_A |0+1+2+3\rangle_B |0+1+2+3+4\rangle_C|000\rangle_{abc}\\
		&~~~~~~~~~~~~+|2\rangle_A |0+1+2+3\rangle_B |0+1+2+3+4\rangle_C|111\rangle_{abc}, \\
		&|\phi_8\rangle \to  |1\rangle_A |0+1\rangle_B |0-1\rangle_C|000\rangle_{abc}, \\
		&|\phi_9\rangle \to  |2\rangle_A |0+1\rangle_B |0-2\rangle_C|111\rangle_{abc}, \\
		&|\phi_{10}\rangle \to  |0-1\rangle_A |0+1\rangle_B |2\rangle_C|000\rangle_{abc}, \\
		&|\phi_{11}\rangle \to  |0-1\rangle_A |0+1\rangle_B |3\rangle_C|000\rangle_{abc}, \\
		&|\phi_{12}\rangle \to  |0-1\rangle_A |0+1\rangle_B |4\rangle_C|000\rangle_{abc}, \\
		&|\phi_{13}\rangle \to  |0+1\rangle_A |0+1\rangle_B \bigl(|2\rangle_C + \omega_3 |3\rangle_C + \omega_3^2 |4\rangle_C\bigr)|000\rangle_{abc}, \\
		&|\phi_{14}\rangle \to  |0+1\rangle_A |0+1\rangle_B \bigl(|2\rangle_C + \omega_3^2 |3\rangle_C + \omega_3 |4\rangle_C\bigr)|000\rangle_{abc}.
	\end{aligned}
	\]
	
	\textit{Step 2.} Bob performs the measurement
	\[
	\begin{aligned}
		\mathcal{M}_2\equiv \bigl\{
		M_{2,1}=P[(|2\rangle,|3\rangle)_{B};|0\rangle_{b}],
		\;
		M_{2,2}=I-M_{2,1}
		\bigr\}.
	\end{aligned}
	\]
	If outcome \(M_{2,1}\) occurs, then the states are \(|\phi_{3} \rangle,~|\phi_{5} \rangle,~|\phi_{6} \rangle,\) and $\left |\phi_{7}  \right \rangle \to|0+1\rangle_A |2+3\rangle_B |0+1+2+3+4\rangle_C|000\rangle_{abc}$. Next,
	Alice performs the measurement
	\[
	\begin{aligned}
		\mathcal{M}_3 \equiv \bigl\{
		M_{3,1}=P\bigl[|0-1\rangle_A;|0\rangle_a\bigr],\;
		M_{3,2}=I-M_{3,1}\bigr\}.
	\end{aligned}
	\]If outcome \(M_{3,1}\) occurs, then the states are \(|\phi_{3} \rangle,|\phi_{5} \rangle\); if outcome \(M_{3,2}\) occurs, then the states are  \(|\phi_{6} \rangle\) and $\left |\phi_{7}  \right \rangle \to|0+1\rangle_A |2+3\rangle_B |0+1+2+3+4\rangle_C|000\rangle_{abc}$. They are all locally distinguishable. If outcome \(M_{2,2}\) occurs, we  move to the next step.
	
	\textit{Step 3.} Charlie performs the measurement
	\[
	\begin{aligned}
		\mathcal{M}_4\equiv \bigl\{
		M_{4,1}=P[(|2\rangle,|3\rangle,|4\rangle)_{C};|0\rangle_{c}],
		\;
		M_{4,2}=I-M_{4,1}
		\bigr\}.
	\end{aligned}
	\]
	If outcome \(M_{4,1}\) occurs, then the state is one of \(\{|\phi_{10} \rangle,|\phi_{11} \rangle,|\phi_{12} \rangle,|\phi_{13} \rangle,|\phi_{14} \rangle, \left |\phi_{7}   \right \rangle \to|0+1\rangle_A |0+1\rangle_B |2+3+4\rangle_C|000\rangle_{abc}\}\). Then, Alice again performs\[
	\begin{aligned}
		\mathcal{M}_5 \equiv \bigl\{
		M_{5,1}=P\bigl[|0-1\rangle_A;|0\rangle_a\bigr],\;
		M_{5,2}=I-M_{5,1}\bigr\}.
	\end{aligned}
	\]If outcome \(M_{5,1}\) occurs, the state is one of \(\{|\phi_{10} \rangle,|\phi_{11} \rangle,|\phi_{12} \rangle\}\), which can be perfectly distinguished by measuring subsystem $C$. If outcome \(M_{5,2}\) occurs, then the remaining possibilities are  \(|\phi_{13} \rangle,~|\phi_{14} \rangle,\) and $\left |\phi_{7}  \right \rangle\to|0+1\rangle_A |0+1\rangle_B |2+3+4\rangle_C|000\rangle_{abc}$, which are locally distinguishable.
	
	\textit{Step 4.} Alice then performs the measurement
	\[
	\begin{aligned}
		\mathcal{M}_6\equiv \bigl\{
		M_{6,1}=P[|1\rangle_A;|0\rangle_a],
		\;
		M_{6,2}=I-M_{6,1}
		\bigr\}.
	\end{aligned}
	\]
	If outcome \(M_{6,1}\) occurs, the state is one of \(\{|\phi_{1} \rangle,|\phi_{8} \rangle,\left |\phi_{7}  \right  \rangle \to|1\rangle_{A}|0+1\rangle_{B}|0+1\rangle_C|000\rangle_{abc}\}\), which is locally distinguishable. Otherwise, we proceed to the next step.
	
	\textit{Step 5.} Charlie performs the measurement
	\[
	\begin{aligned}
		\mathcal{M}_7\equiv \bigl\{
		M_{7,1}=P[|0-2\rangle_C;|1\rangle_c],
		\;
		M_{7,2}=I-M_{7,1}
		\bigr\}.
	\end{aligned}
	\]
	If outcome \(M_{7,1}\) occurs, then the state is  \(|\phi_{9} \rangle\).
	If outcome \(M_{7,2}\) occurs, then the state is one of  \(\{|\phi_{2} \rangle,|\phi_{4} \rangle, \left |\phi_{7}  \right  \rangle \to|0\rangle_A |0+1\rangle_B |0+1\rangle_C|000\rangle_{abc}+|2\rangle_A |0+1+2+3\rangle_B |0+1+2+3+4\rangle_C|111\rangle_{abc}\}\), which is locally distinguishable.
	
	If other outcome occurs in Step 1, similar protocols can be constructed to perfectly distinguish the corresponding states  by LOCC.
	\qed
	
	The genuinely nonlocal set of type II is locally irreducible, but it bears structural similarities to the previously discussed sets. For the states in the subset $\{|\phi_i\rangle\}_{i=1}^{7}$ given by Eq. (\ref{GNII1}), the reduced density operators on $A$ and $B$ exhibit the same structural feature as the set in Eq. (\ref{N2}) of bipartite system, and the states $|\phi_i\rangle$ ($i=1,2,\ldots,6$) on Charlie's subsystem is fixed as $|1\rangle$. For the states in the subset $\{|\phi_4\rangle\}\bigcup\{|\phi_i\rangle\}_{i=7}^{14}$ (\ref{GNII1}), the same bipartite structural form appears between Alice and Charlie, whereas the states $|\phi_i\rangle$ ($i=8,9,\ldots,14$) on Bob's subsystem is $|0+1\rangle$. These states are locally indistinguishable. Moreover, only one EPR pair is insufficient to complete the locally distinguishable task of set in (\ref{GNII1}). Therefore, in this protocol, we let the subsystems share the GHZ state and find that a single copy is enough.
	
	In $\mathbb{C}^{d_1} \otimes \mathbb{C}^{d_2} \otimes \mathbb{C}^{d_3}$ ($3 \leq d_1 \leq d_2 \leq d_3$), the genuinely nonlocal set of type II in Ref~\cite{Lu2024GenuinelyNS} is
	\begin{equation}\label{GNII2}
		\begin{aligned}
			& |\phi_i\rangle = |i\rangle_A |0-i\rangle_B |1\rangle_C, \quad 1 \leq i \leq d_1-1, \\
			& |\phi_{i+d_1-1}\rangle = |0-i\rangle_A |j\rangle_B |1\rangle_C, \quad 1 \leq i \leq d_1-2, \ j = i+1; \ i = d_1-1, \ j=1, \\
			& |\phi_{m+d_1-1}\rangle = |0-1\rangle_A |m\rangle_B |1\rangle_C, \quad d_1 \leq m \leq d_2-1, \\
			& |\phi_{s_1+d_1+d_2-2}\rangle = |0+1\rangle_A \left( |2\rangle_B + \sum_{t_1=1}^{d_2-d_1} \omega_{d_2-d_1+1}^{s_1 t_1} |t_1+d_1-1\rangle_B \right) |1\rangle_C, \quad 1 \leq s_1 \leq d_2-d_1, \\
			& |\phi_{2d_2-1}\rangle = |0+1+\cdots+(d_1-1)\rangle_A |0+1+\cdots+(d_2-1)\rangle_B|0+1+\cdots+(d_3-1)\rangle_C, \\
			& |\phi_{i+2d_2-1}\rangle = |i\rangle_A |0+1\rangle_B |0-i\rangle_C, \quad 1 \leq i \leq d_1-1, \\
			& |\phi_{i+d_1+2d_2-2}\rangle = |0-i\rangle_A |0+1\rangle_B |j\rangle_C, \quad 1 \leq i \leq d_1-2, \ j=i+1, \\
			& |\phi_{n+d_1+2d_2-3}\rangle = |0-1\rangle_A |0+1\rangle_B |n\rangle_C, \quad d_1 \leq n \leq d_3-1, \\
			& |\phi_{s_2+d_1+2d_2+d_3-4}\rangle = |0+1\rangle_A |0+1\rangle_B \left( |2\rangle_C + \sum_{t_2=1}^{d_3-d_1} \omega_{d_3-d_1+1}^{s_2 t_2} |t_2+d_1-1\rangle_C \right), \quad 1 \leq s_2 \leq d_3-d_1.
		\end{aligned}
	\end{equation}
	
	\begin{theorem}\label{thm:GNII2}
		The set in Eq.(\ref{GNII2}) can be locally distinguished by using the resource configuration
		$(1, |G(3)\rangle_{ABC})$.
	\end{theorem}
	
	Since set (\ref{GNII2}) is a higher-dimensional generalization of the previous set and shares the same structural features, we establish similar method in this protocol. The detailed procedure is given in Appendix~\ref{D}. For the scheme in Theorem \ref{thm:GNII2}, we use a total of one maximally entangled states in three-qutrit quantum system.
	
	\subsection{ Genuinely nonlocal sets in $\otimes_{i=1}^{n}\mathbb{C}^{d_i}$ ($3 \leq d_1 \leq d_2 \leq \dots \leq d_n, n \geq 4$)}
	In $\otimes_{i=1}^{n}\mathbb{C}^{d_i}$ ($3 \leq d_1 \leq d_2 \leq \dots \leq d_n, n \geq 4$),
	the genuinely nonlocal set of type II~\cite{Lu2024GenuinelyNS} is
	\begin{equation}\label{GNII3}
		\begin{aligned}
			&G_1 = \bigl\{|x_{j_1}\rangle_{A_1}|y_{j_1}\rangle_{A_2}|0\rangle_{A_3}|0\rangle_{A_4}\cdots |0\rangle_{A_{n-3}}|0\rangle_{A_{n-2}}|0\rangle_{A_{n-1}}|1\rangle_{A_n}\bigr\}, \\
			&G_2 = \bigl\{|1\rangle_{A_1}|x_{j_2}\rangle_{A_2}|z_{j_2}\rangle_{A_3}|0\rangle_{A_4}\cdots |0\rangle_{A_{n-3}}|0\rangle_{A_{n-2}}|0\rangle_{A_{n-1}}|0\rangle_{A_n}\bigr\}, \\
			&G_3 = \bigl\{|0\rangle_{A_1}|x_{j_3}\rangle_{A_2}|1\rangle_{A_3}|z_{j_3}\rangle_{A_4}\cdots |0\rangle_{A_{n-3}}|0\rangle_{A_{n-2}}|0\rangle_{A_{n-1}}|0\rangle_{A_n}\bigr\}, \\
			&~~\vdots \\
			&G_{n-2} = \bigl\{|0\rangle_{A_1}|x_{j_{n-2}}\rangle_{A_2}|0\rangle_{A_3}|0\rangle_{A_4}\cdots |0\rangle_{A_{n-3}}|1\rangle_{A_{n-2}}|z_{j_{n-2}}\rangle_{A_{n-1}}|0\rangle_{A_n}\bigr\}, \\
			&G_{n-1} = \bigl\{|0\rangle_{A_1}|x_{j_{n-1}}\rangle_{A_2}|0\rangle_{A_3}|0\rangle_{A_4}\cdots |0\rangle_{A_{n-3}}|0\rangle_{A_{n-2}}|1\rangle_{A_{n-1}}|z_{j_{n-1}}\rangle_{A_n}\bigr\},
		\end{aligned}
	\end{equation}
	where $\{|x_{j_1}\rangle|y_{j_1}\rangle, m_1 = 1, 2, \dots, 2d_2-1\}$ are $2d_2-1$ locally indistinguishable orthogonal product states in $\mathbb{C}^{d_1} \otimes \mathbb{C}^{d_2}$. And $\{|x_{j_{k-1}}\rangle|z_{j_{k-1}}\rangle, j_{k-1} = 1, 2, \dots, 2d_k-1\}$ are $2d_k-1$ locally indistinguishable orthogonal product states in $\mathbb{C}^{d_2} \otimes \mathbb{C}^{d_k}$ ($k = 3, \dots, n$).
	
	\begin{theorem}\label{thm:GNII3}
		The genuinely nonlocal set $\bigcup_{i=1}^{n-1} G_i$ given by Eq. (\ref{GNII3}) can be locally distinguished by using the resource configuration
		\(\{(1,\vert \phi^{+}(2)\rangle_{A_{n-1}A_n});(\frac{r_1}{s},\vert \phi^{+}(3)\rangle_{A_1A_2});(\frac{r_2}{s},\vert \phi^{+}(3)\rangle_{A_2A_3});(\frac{r_3}{s},\vert \phi^{+}(3)\rangle_{A_2A_4});\cdots;(\frac{r_{n-1}}{s},\vert \phi^{+}(3)\rangle_{A_2A_n})\}\),  where \[
		\begin{aligned}
			s= \sum_{k=2}^{n} \left ( 2d_k-1 \right )  ,\;
			r_i=2 d_{i+1}-1.
		\end{aligned}
		\]
	\end{theorem}
	
	Unlike set (\ref{GNI3}), due to its local irreducibility, the genuinely nonlocal set in (\ref{GNII3}) is impossible to locally distinguish each subset \(G_i\) (\( i = 1,2,\ldots,n-1 \)). So we let $A_{}n-1$ and $A_n$ share an EPR state. Then we can identify all the $G_i$ by LOCC. After that, by applying steps similar to those in Appendix~\ref{A}, we can perfectly discriminate the states in each subset \(G_i\) (\( i = 1,2,\ldots,n-1 \). The detailed procedure is given in Appendix~\ref{E}. Compared with the discrimination protocol of genuinely nonlocal sets (\ref{GNI3}) of type I, the protocol in Theorem \ref{thm:GNII3} consumes one more ebit of entanglement resource.

\section{CONCLUSION}\label{Q6}

We have studied entanglement-assisted local discrimination for two classes of genuinely nonlocal sets: type I (locally reducible) and type II (locally irreducible) \cite{Lu2024GenuinelyNS}. These sets are constructed based on nonlocal sets in bipartite systems. We proposed efficient local discrimination schemes for various orthogonal product sets. Among the eight presented schemes, only the final one utilizes two maximally entangled states, while the remaining seven require only one. In $\mathbb{C}^{d_1} \otimes \mathbb{C}^{d_2}$ ($3 \le d_1 \le d_2$) systems, the nonlocal sets can be perfectly discriminated using at most one $|\phi^+(3)\rangle$ state; when $d_1$ and $d_2$ are relatively small, an EPR state suffices.

Building on these bipartite conclusions, we provided locally distinguishable protocols for genuinely nonlocal sets in multipartite systems. For type I sets, the results align with their bipartite counterparts. For type II sets, one GHZ state is sufficient for low-dimensional tripartite systems, while a maximally entangled state $|G(3)\rangle = \frac{1}{\sqrt{3}}(|000\rangle + |111\rangle + |222\rangle)$ is required for higher-dimensional three-qutrit systems. In $n$-partite systems, the discrimination protocol for type II sets requires only one additional EPR state compared to type I sets. Overall, our findings demonstrate how auxiliary subsystems can be effectively exploited to design efficient entanglement-assisted discrimination protocols for genuinely nonlocal sets, providing favorable theoretical support for quantum information processing.

\begin{acknowledgments}
	This work was supported by the National Natural Science Foundation of China (Grant No.~12526564), the Hebei Natural Science Foundation (Grant No.~A2025403008), and the Doctoral Science Start Foundation of Hebei GEO University (Grant No.~BQ2024075).
\end{acknowledgments}
\begin{appendix}
\section{The proof of Theorem 2}\label{A}
		Let Alice and Bob share a maximally entangled state \(\vert \phi^{+}(3)\rangle_{ab}\). The initial state is
		\[
		\begin{aligned}
			\left|\phi\right\rangle_{AB} \otimes \left|\phi^{+}(3)\right\rangle_{ab},
		\end{aligned}
		\]
		where $a$ and $b$ are the ancillary systems of Alice and Bob, respectively. Because each of subsets $\{|\phi_i\rangle\}$, $\{|\phi_{i+(d_1-1)}\rangle\}$, $\{|\phi_{j+(d_1-1)}\rangle\}$, and $\{|\phi_{s+(d_1+d_2-2)}\rangle\}$ is LOCC distinguishable, we can also locally distinguish these subsets. The protocol proceeds as follows.
		
		\textbf{Step 1.} Alice performs the measurement
		\[
		\begin{aligned}
			\mathcal{M}_1 \equiv \bigl\{
			&M_{1,1}=P\bigl[(|0\rangle,|1\rangle)_A;|0\rangle_a\bigr]
			+P\bigl[|2\rangle_A;|1\rangle_a\bigr]+P\bigl[(|3\rangle,|4\rangle,\dots,|d_1-1\rangle)_A;|2\rangle_a\bigr], \\
			\;
			&M_{1,2}=P\bigl[(|0\rangle,|1\rangle)_A;|1\rangle_a\bigr]
			+P\bigl[|2\rangle_A;|0\rangle_a\bigr]+P\bigl[(|3\rangle,|4\rangle,\dots,|d_1-1\rangle)_A;|2\rangle_a\bigr], \\
			\;
			&M_{1,3}=I-M_{1,1}-M_{1,2}\bigr\}.
		\end{aligned}
		\]
		Suppose that the outcome corresponding to \(M_{1,1}\) occurs. When $d_1$ and $d_2$ are sufficiently large, i.e. $d_1\geq 5$ and $d_2>d_1$, we have
		\[
		\begin{aligned}
			&|\phi_1\rangle \to  |1\rangle_A |0-1\rangle_B|00\rangle_{ab},  \\
			&|\phi_2\rangle \to  |2\rangle_A |0-2\rangle_B|11\rangle_{ab},  \\
			&|\phi_i\rangle \to  |i\rangle_A |0-i\rangle_B|22\rangle_{ab}, \quad  i =3,4,\dots, d_1-1, \\
			&|\phi_{1+(d_1-1)}\rangle \to  |0-1\rangle_A |2\rangle_B|00\rangle_{ab} ,\\
			&|\phi_{2+(d_1-1)}\rangle \to  |0\rangle_A |3\rangle_B|00\rangle_{ab}-|2\rangle_A |3\rangle_B|11\rangle_{ab}, \\
			&|\phi_{i+(d_1-1)}\rangle \to  |0\rangle_A |m\rangle_B|00\rangle_{ab}-|i\rangle_A |m\rangle_B|22\rangle_{ab}, \quad  i =3,4,\dots,d_1-2, \ m = i+1, \\
			&|\phi_{(d_1-1)+(d_1-1)}\rangle \to  |0\rangle_A |1\rangle_B|00\rangle_{ab}-|d_1-1\rangle_A |1\rangle_B|22\rangle_{ab}, \\
			&|\phi_{j+(d_1-1)}\rangle \to  |0-1\rangle_A |j\rangle_B|00\rangle_{ab}, \quad  j =d_1,\dots, d_2-1, \\
			&|\phi_{s+(d_1+d_2-2)}\rangle \to |0+1\rangle_A \Bigl(|2\rangle_B + \sum_{t=1}^{d_2-d_1} \omega_{d_2-d_1+1}^{st} |t+d_1-1\rangle_B \Bigr)|00\rangle_{ab},
			\quad  s =1,\dots, d_2-d_1.\\
		\end{aligned}
		\]

		\textbf{Step 2.} Bob performs the measurement
		\[
		\begin{aligned}
			\mathcal{M}_2\equiv \bigl\{
			M_{2,1}=P[(|2\rangle,|d_1\rangle,\dots,|d_2-1\rangle)_B;|0\rangle_{b}],
			\;
			M_{2,2}=I-M_{2,1}
			\bigr\}.
		\end{aligned}
		\]
		If outcome \(M_{2,1}\) occurs, then the subset is one of  \(\{|\phi_{1+(d_1-1)}\rangle\}\), \(\{|\phi_{j+(d_1-1)}\rangle\}\), and \(\{|\phi_{s+(d_1+d_2-2)}\rangle\}\). They are locally distinguishable.
		
		\textbf{Step 3.} Alice performs the measurement
		\[
		\begin{aligned}
			\mathcal{M}_3\equiv \bigl\{
			M_{3,1}=P[|1\rangle_A;|0\rangle_{a}],	\; M_{3,2}=I-M_{3,1}
			\bigr\}.
		\end{aligned}
		\]
		If outcome \(M_{3,1}\) occurs, then the subset is \(|\phi_{1}\rangle\). Otherwise, the state belongs to the set of remaining states \(\{|\phi_{i}\rangle,|\phi_{i+(d_1-1)}\rangle\}_{i=2}^{d_1-1}\).
		
		\textbf{Step 4.} Bob performs the measurement
		\[
		\begin{aligned}
			\mathcal{M}_4\equiv \bigl\{
			&M_{4,1}=P[|1\rangle_B;(|0\rangle,|2\rangle)_{b}],\\	
			&M_{4,2}=P[(|0\rangle,|2\rangle)_B;|1\rangle_{b}], \\
			\;
			&M_{4,3}=P[|3\rangle_B;(|0\rangle,|1\rangle)_{b}], \\
			\;
			&M_{4,4}=I-M_{4,1}-M_{4,2}-M_{4,3}
			\bigr\}.
		\end{aligned}
		\]
		If outcome \(M_{4,1}\) occurs, the state is \(|\phi_{(d_1-1)+(d_1-1)}\rangle \); if outcome \(M_{4,2}\) occurs, the state is \(|\phi_{2}\rangle\); if outcome \(M_{4,3}\) occurs, the state is \(|\phi_{2+(d_1-1)}\rangle\). Otherwise, the state belongs to the set of remaining states \(\{|\phi_{i}\rangle\}_{i=3}^{d_1-1}\bigcup\{|\phi_{i+(d_1-1)}\rangle\}_{i=3}^{d_1-2}\).
		
		\textbf{Step 5.} Alice performs the measurement
		\[
		\begin{aligned}
			\mathcal{M}_5\equiv \bigl\{
			M_{5,1}=P[|d_1-1\rangle_A;|2\rangle_{a}],	\; M_{5,2}=I-M_{5,1}
			\bigr\}.
		\end{aligned}
		\]
		If outcome \(M_{5,1}\) occurs, then the state is \(|\phi_{d_1-1}\rangle \). Otherwise, Bob performs the measurement
		\[
		\begin{aligned}
			\mathcal{M}_6\equiv \bigl\{
			M_{6,1}=P[|d_1-1\rangle_B;(|0\rangle,|2\rangle)_{b}],	\; M_{6,2}=I-M_{6,1}
			\bigr\}.
		\end{aligned}
		\]
		If outcome \(M_{6,1}\) occurs, then the state is \(|\phi_{(d_1-2)+(d_1-1)}\rangle \). Operate alternately in this way. If outcome \(M_{6,2}\) occurs, Alice performs the measurement
		\[
		\begin{aligned}
			\mathcal{M}_7\equiv \bigl\{
			M_{7,1}=P[|d_1-2\rangle_A;|2\rangle_{a}],	\; M_{7,2}=I-M_{7,1}
			\bigr\}.
		\end{aligned}
		\]
		If outcome \(M_{7,1}\) occurs, then the state is \(|\phi_{d_1-2}\rangle \). Otherwise, Bob performs the measurement
		\[
		\begin{aligned}
			\mathcal{M}_8\equiv \bigl\{
			M_{8,1}=P[|d_1-2\rangle_B;(|0\rangle,|2\rangle)_{b}],	\; M_{8,2}=I-M_{8,1}
			\bigr\}.
		\end{aligned}
		\]
		If outcome \(M_{8,1}\) occurs, then the state is \(|\phi_{(d_1-3)+(d_1-1)}\rangle \).
		
		By the same reasoning, the sets \(\{|\phi_{i}\rangle\}_{i=3}^{d_1-1}\bigcup\{|\phi_{i+(d_1-1)}\rangle\}_{i=3}^{d_1-2}\) can be perfectly distinguished.
		
		In Step 1, when $d_1$ and $d_2$ are relatively small, the number of quantum states in set (\ref{N2}) decreases accordingly, while the structure of the set remains unchanged. A maximally entangled state suffices to perfectly discriminate the larger set, and it can also fully distinguish the smaller one. In fact, when $d_1$ and $d_2$ are relatively small, an EPR state is adequate for the locally distinguishable task, and Theorem \ref{thm:n1} serves as such an example. If any other outcomes occur in Step 1, similar protocols can be constructed to perfectly distinguish the corresponding states.

		\section{The proof of Theorem 4}\label{B}
		
		Because each of subsets $\{|\phi_i\rangle\}$, $\{|\phi_{i+(d_1-2)}\rangle\}$, $\{|\phi_{j+(d_1-2)}\rangle\}$, $\{|\phi_{s+(d_1+d_2-3)}\rangle\}$, $\{|\phi_{2d_2-1}\rangle\}$, $\{|\phi_{i+2d_2-1}\rangle\}$, $\{|\phi_{i+3d_2-2}\rangle\}$, $\{|\phi_{j+3d_2-2}\rangle\}$, $\{|\phi_{s+3d_2+d_3-3}\rangle\}$, and $\{|\phi_{2(d_2+d_3)-2}\rangle\}$ in Eq. (\ref{GNI2}) is LOCC distinguishable, we can also locally distinguish these subsets. The protocol proceeds as follows.
		
		\textbf{Step 1.} Alice performs the measurements\[
		\begin{aligned}
			\mathcal{M}_1 \equiv \bigl\{
			M_{1,1}=|d_1'-1\rangle_A \langle d_1'-1|,\;
			M_{1,2}=I-M_{1,1}\bigr\},
		\end{aligned}
		\]the whole set can be locally reduced to two disjoint subsets. The corresponding results for \(M_{1,1}\) and \(M_{1,2}\) are \(\left \{ \left | \phi _{i}   \right \rangle  \right \} _{i=2d_2}^{  2\left (d_2+d_3 \right )-2 }\) and \(\left \{ \left | \phi _{i}   \right \rangle  \right \} _{i=1}^{2d_2-1}\), respectively. Similar to the case in Appendix \ref{A}, quantum states $|\phi_{2d_2-1}\rangle$ and $|\phi_{2(d_2+d_3)-2}\rangle$ will not be discussed in the subsequent protocols.
		
		\textbf{Step 2.} For the branch corresponding to \(M_{1,2}\), Alice and Bob share a maximally entangled state \(\vert \phi^{+}(3)\rangle_{ab_1}\). The initial state is
		\[
		\begin{aligned}
			\left|\phi\right\rangle_{ABC} \otimes \left|\phi^{+}(3)\right\rangle_{ab_1},
		\end{aligned}
		\] where $a$ and $b_1$ are the ancillary systems of Alice and Bob, respectively. Then Alice performs the measurement
		\[
		\begin{aligned}
			\mathcal{M}_2 \equiv \bigl\{
			&M_{2,1}=P\bigl[(|0\rangle,|1\rangle)_A;|0\rangle_a\bigr]
			+P\bigl[|2\rangle_A;|1\rangle_a\bigr]+P\bigl[(|3\rangle,|4\rangle,\dots,|d_1-1\rangle)_A;|2\rangle_a\bigr], \\
			\;
			&M_{2,2}=P\bigl[(|0\rangle,|1\rangle)_A;|1\rangle_a\bigr]
			+P\bigl[|2\rangle_A;|0\rangle_a\bigr]+P\bigl[(|3\rangle,|4\rangle,\dots,|d_1-1\rangle)_A;|2\rangle_a\bigr], \\
			\;
			&M_{2,3}=I-M_{2,1}-M_{2,2}\bigr\}.
		\end{aligned}
		\]
		Suppose that the outcome corresponding to \(M_{2,1}\) occurs. Then
		\[
		\begin{aligned}
			&|\phi_1\rangle \to |1\rangle_A |0-1\rangle_B |0\rangle_C|00\rangle_{ab_1}, \\
			&|\phi_2\rangle \to |2\rangle_A |0-2\rangle_B |0\rangle_C|11\rangle_{ab_1}, \\
			&|\phi_i\rangle \to  |i\rangle_A |0-i\rangle_B |0\rangle_C|22\rangle_{ab_1}, \quad  i =3,4,\dots, d_1-2,  \\
			&|\phi_{1+(d_1-2)}\rangle \to  |0-1\rangle_A |2\rangle_B |0\rangle_C|00\rangle_{ab_1}, \\
			&|\phi_{2+(d_1-2)}\rangle \to  |0\rangle_A |3\rangle_B |0\rangle_C|00\rangle_{ab_1}-|2\rangle_A |3\rangle_B |0\rangle_C|11\rangle_{ab_1}, \\
			&|\phi_{i+(d_1-2)}\rangle \to  |0\rangle_A |m\rangle_B |0\rangle_C|00\rangle_{ab_1}-|i\rangle_A |m\rangle_B |0\rangle_C|22\rangle_{ab_1}, \quad  i =3,4,\dots, d_1-3,  \quad m = i+1, \\
			&|\phi_{(d_1-2)+(d_1-2)}\rangle \to  |0\rangle_A |1\rangle_B |0\rangle_C|00\rangle_{ab_1}-|(d_1-2)\rangle_A |1\rangle_B |0\rangle_C|22\rangle_{ab_1}, \\
			&|\phi_{j+(d_1-2)}\rangle \to  |0-1\rangle_A |j\rangle_B |0\rangle_C|00\rangle_{ab_1}, \quad  j =d_1-1, \dots, d_2-1, \\
			&|\phi_{s+(d_1+d_2-3)}\rangle \to  |0+1\rangle_A |\eta\rangle_B |0\rangle_C|00\rangle_{ab_1}, \quad  s = 1 , \dots,  d_2-d_1+1,
		\end{aligned}
		\]	where $|\eta\rangle_B = |2\rangle_B + \sum_{t=1}^{d_2-d_1+1} \omega_{d_2-d_1+2}^{st} |t+d_1-2\rangle_B$.
		
		\textbf{Step 3.} Bob performs the measurement
		\[
		\begin{aligned}
			\mathcal{M}_3\equiv \bigl\{
			M_{3,1}=P[(|2\rangle,|d_1\rangle,\dots,|d_2-1\rangle)_B;|0\rangle_{b_1}],
			\;
			M_{3,2}=I-M_{3,1}
			\bigr\}.
		\end{aligned}
		\]
		If outcome \(M_{3,1}\) occurs, then the subset is one of  \(\{|\phi_{1+(d_1-1)}\rangle\}\), \(\{|\phi_{j+(d_1-1)}\rangle\}\), and \(\{ |\phi_{s+(d_1+d_2-2)}\rangle\}\),  which are locally distinguishable. Otherwise, the state belongs to one of the remaining subsets.
		
		\textbf{Step 4.} Alice performs the measurement
		\[
		\begin{aligned}
			\mathcal{M}_4\equiv \bigl\{
			M_{4,1}=P[|1\rangle_A;|0\rangle_{a}],	\; M_{4,2}=I-M_{4,1}
			\bigr\}.
		\end{aligned}
		\]
		If outcome \(M_{4,1}\) occurs, then the state is \(|\phi_{1}\rangle\). Otherwise, we proceed to the next step.
		
		\textbf{Step 5.} Bob performs the measurement
		\[
		\begin{aligned}
			\mathcal{M}_5\equiv \bigl\{
			&M_{5,1}=P[|1\rangle_B;(|0\rangle,|2\rangle)_{b_1}],\\	
			&M_{5,2}=P[(|0\rangle,|2\rangle)_B;|1\rangle_{b_1}], \\
			\;
			&M_{5,3}=P[|3\rangle_B;(|0\rangle,|1\rangle)_{b_1}], \\
			\;
			&M_{5,4}=I-M_{4,1}-M_{4,2}-M_{4,3}
			\bigr\}.
		\end{aligned}
		\]
		If outcome \(M_{5,1}\) occurs, then the state is \(|\phi_{(d_1-2)+(d_1-2)}\rangle \); if outcome \(M_{5,2}\) occurs, then the state is \(|\phi_{2}\rangle\); if outcome \(M_{5,3}\) occurs, then the state is \(|\phi_{2+(d_1-2)}\rangle\). Otherwise, the state belongs to the set of remaining states \(\{|\phi_{i}\rangle\}_{i =3}^{d_1-2}\bigcup\{|\phi_{i+(d_1-2)}\rangle\}_{i =3}^{d_1-3}\).
		
		\textbf{Step 6.} Alice performs the measurement
		\[
		\begin{aligned}
			\mathcal{M}_6\equiv \bigl\{
			M_{6,1}=P[|d_1-2\rangle_A;|2\rangle_{a}],	\; M_{6,2}=I-M_{6,1}
			\bigr\}.
		\end{aligned}
		\]
		If outcome \(M_{6,1}\) occurs, then the state is \(|\phi_{d_1-2}\rangle \). Otherwise, Bob performs the measurement
		\[
		\begin{aligned}
			\mathcal{M}_7\equiv \bigl\{
			M_{7,1}=P[|d_1-2\rangle_B;(|0\rangle,|2\rangle)_{b_1}],	\; M_{7,2}=I-M_{7,1}
			\bigr\}.
		\end{aligned}
		\]
		If outcome \(M_{7,1}\) occurs, then the state is \(|\phi_{(d_1-3)+(d_1-2)}\rangle \). Operate alternately in this way, the set \(\{|\phi_{i}\rangle\}_{i =3}^{d_1-2}\bigcup\{|\phi_{i+(d_1-2)}\rangle\}_{i =3}^{d_1-3}\) can be perfectly distinguished.
		
		\textbf{Step 2$'$.} For the branch corresponding to  \(M_{1,1}\), Alice and Bob share a maximally entangled state \(\vert \phi^{+}(3)\rangle_{b_2c}\). The initial state is
		\[
		\begin{aligned}
			\left|\phi\right\rangle_{ABC} \otimes \left|\phi^{+}(3)\right\rangle_{b_2c},
		\end{aligned}
		\] where $b_2$ and $c$ are the ancillary systems of Bob and Charlie, respectively. Then Bob performs the measurement
		\[
		\begin{aligned}
			\mathcal{M}_{8} \equiv \bigl\{
			&M_{8,1}=P\bigl[(|0\rangle,|1\rangle)_B;|0\rangle_{b_2}\bigr]
			+P\bigl[|2\rangle_B;|1\rangle_{b_2}\bigr]+P\bigl[(|3\rangle,|4\rangle,\dots,|d_1-1\rangle)_B;|2\rangle_{b_2}\bigr], \\
			\;
			&M_{8,2}=P\bigl[(|0\rangle,|1\rangle)_B;|1\rangle_{b_2}\bigr]
			+P\bigl[|2\rangle_B;|0\rangle_{b_2}\bigr]+P\bigl[(|3\rangle,|4\rangle,\dots,|d_1-1\rangle)_B;|2\rangle_{b_2}\bigr], \\
			\;
			&M_{8,3}=I-M_{8,1}-M_{8,2}\bigr\}.
		\end{aligned}
		\]
		Suppose that the outcome corresponding to \(M_{8,1}\) occurs. Then
		\[
		\begin{aligned}
			&|\phi_{1+2d_2-1}\rangle \to  |d_1-1\rangle_A |1\rangle_B |0-1\rangle_C|00\rangle_{b_2c}, \\
			&|\phi_{2+2d_2-1}\rangle \to  |d_1-1\rangle_A |2\rangle_B |0-2\rangle_C|11\rangle_{b_2c}, \\
			&|\phi_{i+2d_2-1}\rangle \to  |d_1-1\rangle_A |i\rangle_B |0-i\rangle_C|22\rangle_{b_2c}, \quad  i = 3, 4, \dots, d_2-1,  \\	
			&|\phi_{1+3d_2-2}\rangle \to  |d_1-1\rangle_A |0-1\rangle_B |2\rangle_C|00\rangle_{b_2c}, \\
			&|\phi_{2+3d_2-2}\rangle \to  |d_1-1\rangle_A |0\rangle_B |3\rangle_C|00\rangle_{b_2c}-|d_1-1\rangle_A |2\rangle_B |3\rangle_C|11\rangle_{b_2c}, \\
			&|\phi_{i+3d_2-2}\rangle \to  |d_1-1\rangle_A |0\rangle_B |m\rangle_C|00\rangle_{b_2c}-|d_1-1\rangle_A |i\rangle_B |m\rangle_C|22\rangle_{b_2c}, \quad  i = 3, 4, \dots, d_2-2, \quad m = i+1, \\
			&|\phi_{(d_2-1)+3d_2-2}\rangle \to  |d_1-1\rangle_A |0\rangle_B |1\rangle_C|00\rangle_{b_2c}-|d_1-1\rangle_A |d_2-1\rangle_B |1\rangle_C|22\rangle_{b_2c}, \\
			&|\phi_{j+3d_2-2}\rangle \to  |d_1-1\rangle_A |0-1\rangle_B |j\rangle_C|00\rangle_{b_2c},\quad j=d_2, \dots, d_3-1, \\
			&|\phi_{s+3d_2+d_3-3}\rangle \to  |d_1-1\rangle_A |0+1\rangle_B |\eta\rangle_C|00\rangle_{b_2c}, \quad  s =1, \dots, d_3-d_2,
		\end{aligned}
		\]where $|\eta\rangle_C = |2\rangle_C + \sum_{t=1}^{d_3-d_2} \omega_{d_3-d_2+1}^{st} |t+d_2-1\rangle_C$.
		
		\textbf{Step 3$'$.} Charlie then performs the measurement
		\[
		\begin{aligned}
			\mathcal{M}_{9}\equiv \bigl\{
			M_{9,1}=P[(|2\rangle,|d_2\rangle,\dots,|d_3-1\rangle)_C;|0\rangle_{c}],
			M_{9,2}=I-M_{9,1}
			\bigr\}.
		\end{aligned}
		\]
		If outcome \(M_{9,1}\) occurs, then the subsets are \(\{|\phi_{1+3d_2-2}\rangle\}\), \(\{|\phi_{j+3d_2-2)}\rangle\}\), and \(\{|\phi_{s+3d_2+d_3-3)}\rangle\}\), which are locally distinguishable. Otherwise, the state belongs to one of the remaining subsets.
		
		\textbf{Step 4$'$.}  Bob performs
		\[
		\mathcal{M}_{10} \equiv \bigl\{
		M_{10,1}=P[|1\rangle_{B};|0\rangle_{b_2}],
		\;
		M_{10,2}=I-M_{10,1}
		\bigr\}.
		\]
		The corresponding result for \(M_{10,1}\) is \(|\phi_{1+2d_2-1} \rangle\). Otherwise, the state is one of the remaining states.
		
		\textbf{Step 5$'$.} Charlie performs the measurement
		\[
		\begin{aligned}
			\mathcal{M}_{11}\equiv \bigl\{
			&M_{11,1}=P[|1\rangle_C;(|0\rangle,|2\rangle)_{c}], \\
			\;
			&M_{11,2}=P[(|0\rangle,|2\rangle)_C;|1\rangle_{c}], \\
			\;
			&M_{11,3}=P[|3\rangle_C;(|0\rangle,|1\rangle)_{c}], \\
			\;
			&M_{11,4}=I-M_{11,1}-M_{11,1}-M_{11,1}
			\bigr\}.
		\end{aligned}
		\]
		If outcome \(M_{11,1}\) occurs, then the state is \(|\phi_{(d_2-1)+3d_2-2}\rangle\); if outcome \(M_{11,2}\) occurs, then the state is \(|\phi_{2+2d_2-1}\rangle\); if outcome \(M_{11,3}\) occurs, then the state is \(|\phi_{2+3d_2-2}\rangle\). Otherwise, the state belongs to the set of the remaining states \(\{|\phi_{i+2d_2-1}\rangle_{i =3}^{d_2-1}\bigcup\{|\phi_{i+3d_2-2}\rangle\}_{i =3}^{d_2-2}\). Similar to the case of Step 5, these states can be perfectly distinguished.	
		
		The discussions concerning dimensions $d_1$, $d_2$ and $d_3$ coincide with the proof of Theorem \ref{thm:n2}. If any other outcomes in Step 2 or any other outcomes in Step 2$'$ occurs, we can find the distinction protocol similarly.
		
		\section{The proof of Theorem 5}\label{C}
		
		The $\sum_{i=2}^{n}(2d_i-1)$ states in Eq. (\ref{GNI3}) are locally indistinguishable in any bipartition and have genuine nonlocality. But it is locally reducible. When $A_1$ performs the measurements
		\[
		\begin{aligned}
			\mathcal{M}_1 \equiv \bigl\{
			M_{1,1}=|d_1-1\rangle_{A_1}\langle d_1-1|,\;
			M_{1,2}=I-M_{1,1}\bigr\},
		\end{aligned}
		\]
		the whole set can be locally reduced to two disjoint subsets
		$\{G_1, \cdots, G_{n-2}\}$ and $\{G_{n-1}\}$.
		
		Next, $A_n$ performs the measurements \[
		\begin{aligned}
			\mathcal{M}_2 \equiv \bigl\{
			M_{2,1}=|1\rangle_{A_n}\langle 1|,\;
			M_{2,2}=I-M_{2,1}\bigr\},
		\end{aligned}
		\]
		the subset $\{G_1, \cdots, G_{n-2}\}$ can be locally reduced to two disjoint subsets
		$\{G_1\}$ and $\{G_{2}, \dots, G_{n-2}\}$. Similarly, when $A_{n-1}$, $A_{n-1}$, $\dots$, $A_{n-1}$ perform the measurements $\mathcal{M}_2 $, respectively, the set $\{G_{2}, \dots, G_{n-2}\}$ can be distinguished one by one.
		
		For the $\{G_1\}$, let $A_1$ and $A_2$ share a maximally entangled state \(\vert \phi^{+}(2)\rangle_{a_1a_2}\). The initial state is
		\[
		\begin{aligned}
			\left|\phi\right\rangle_{A_1A_2\cdots A_n} \otimes \left|\phi^{+}(2)\right\rangle_{a_1a_2},
		\end{aligned}
		\]
		where $a_1$ and $ a_2$ are the ancillary systems of $A_1$ and $ A_2$, respectively. Then, by performing the same discrimination protocol as in Appendix~\ref{A}, these states can be perfectly distinguished.
		
		For the $\{G_2\}$, let $A_2$ and $A_3$ share a maximally entangled state \(\vert \phi^{+}(2)\rangle_{a_2a_3}\). For the $\{G_3\}$, let $A_2$ and $A_4$ share a maximally entangled state \(\vert \phi^{+}(2)\rangle_{a_2a_4}\). By analogy, for $G_{4}$, $\dots$, $G_{n-1}$, perfect discrimination can be achieved by introducing a maximally entangled state respectively and referring to the entanglement-assisted discrimination protocol in Appendix~\ref{A}.

		\section{The proof of Theorem 7}\label{D}
		
		By using the entanglement resource \(|G(3) \rangle_{abc}=|000\rangle_{abc}+|111\rangle_{abc}+|222\rangle_{abc}\), the initial state is
		\[
		\begin{aligned}
			\left|\phi\right\rangle_{ABC} \otimes \left|G(3)\right\rangle_{abc},
		\end{aligned}
		\]where $a$, $b$, and $c$ are the ancillary systems of Alice, Bob, and Charlie, respectively. Because each subset in Eq. (\ref{GNII2}) are LOCC distinguishable, one only needs to locally distinguish these subsets.		
		
		According to the characteristic of set in (\ref{GNII2}), similar to the distinguishing processes of the set in Theorem \ref{N2}, the  state $|\phi_{2d_2-1}\rangle$ will be omited further discussion in the subsequent protocols.
		
		\textbf{Step 1.} Alice performs the measurement
		\[
		\begin{aligned}
			\mathcal{M}_1 \equiv \bigl\{
			&M_{1,1}=P\bigl[(|0\rangle,|1\rangle)_A;|0\rangle_a\bigr]
			+P\bigl[|2\rangle_A;|1\rangle_a\bigr]+P\bigl[(|3\rangle,|4\rangle,\dots,|d_1-1\rangle)_A;|2\rangle_a\bigr], \\
			\;
			&M_{1,2}=P\bigl[(|0\rangle,|1\rangle)_A;|1\rangle_a\bigr]
			+P\bigl[|2\rangle_A;|0\rangle_a\bigr]+P\bigl[(|3\rangle,|4\rangle,\dots,|d_1-1\rangle)_A;|2\rangle_a\bigr], \\
			\;
			&M_{1,3}=I-M_{1,1}-M_{1,2}\bigr\}.
		\end{aligned}
		\]
		Suppose that the outcome corresponding to \(M_{1,1}\) occurs. Then
		\[
		\begin{aligned}
			& |\phi_1\rangle \to  |1\rangle_A |0-1\rangle_B |1\rangle_C|000\rangle_{abc}, \\
			& |\phi_2\rangle \to  |2\rangle_A |0-2\rangle_B |1\rangle_C|111\rangle_{abc}, \\
			& |\phi_i\rangle \to  |i\rangle_A |0-i\rangle_B |1\rangle_C|222\rangle_{abc}, \quad i=3,\dots,d_1-1, \\
			& |\phi_{1+d_1-1}\rangle \to  |0-1\rangle_A |2\rangle_B |1\rangle_C|000\rangle_{abc}, \\
			& |\phi_{2+d_1-1}\rangle \to  |0\rangle_A |3\rangle_B |1\rangle_C|000\rangle_{abc}-|2\rangle_A |3\rangle_B |1\rangle_C|111\rangle_{abc}, \\
			& |\phi_{i+d_1-1}\rangle \to  |0\rangle_A |j\rangle_B |1\rangle_C|000\rangle_{abc}-|i\rangle_A |j\rangle_B |1\rangle_C|222\rangle_{abc}, \quad i=3,\dots, d_1-2, \quad j = i+1, \\
			& |\phi_{(d_1-1)+(d_1-1)}\rangle \to  |0\rangle_A |1\rangle_B |1\rangle_C|000\rangle_{abc}-|d_1-1\rangle_A |1\rangle_B |1\rangle_C|222\rangle_{abc}, \\
			& |\phi_{m+d_1-1}\rangle \to  |0-1\rangle_A |m\rangle_B |1\rangle_C|000\rangle_{abc}, \quad m=d_1,\dots,d_2-1, \\
			& |\phi_{s_1+d_1+d_2-2}\rangle \to  |0+1\rangle_A \left( |2\rangle_B + \sum_{t_1=1}^{d_2-d_1} \omega_{d_2-d_1+1}^{s_1 t_1} |t_1+d_1-1\rangle_B \right) |1\rangle_C|000\rangle_{abc}, \quad s_1=1,\dots,d_2-d_1, \\
			& |\phi_{1+2d_2-1}\rangle \to  |1\rangle_A |0+1\rangle_B |0-1\rangle_C|000\rangle_{abc}, \\
			& |\phi_{2+2d_2-1}\rangle \to  |2\rangle_A |0+1\rangle_B |0-2\rangle_C|111\rangle_{abc}, \\
			& |\phi_{i+2d_2-1}\rangle \to  |i\rangle_A |0+1\rangle_B |0-i\rangle_C|222\rangle_{abc}, \quad  i=3,\dots,d_1-1, \\
			& |\phi_{1+d_1+2d_2-2}\rangle \to  |0-1\rangle_A |0+1\rangle_B |2\rangle_C|000\rangle_{abc}, \\
			& |\phi_{2+d_1+2d_2-2}\rangle \to  |0\rangle_A |0+1\rangle_B |3\rangle_C|000\rangle_{abc}-|2\rangle_A |0+1\rangle_B |3\rangle_C|111\rangle_{abc}, \\
			& |\phi_{i+d_1+2d_2-2}\rangle \to  |0\rangle_A |0+1\rangle_B |j\rangle_C|000\rangle_{abc}-|i\rangle_A |0+1\rangle_B |j\rangle_C|222\rangle_{abc}, \quad  i=3,\dots,d_1-2, \ j=i+1, \\
			& |\phi_{n+d_1+2d_2-3}\rangle \to  |0-1\rangle_A |0+1\rangle_B |n\rangle_C|000\rangle_{abc}, \quad  n=d_1,\dots,d_3-1, \\
			& |\phi_{s_2+d_1+2d_2+d_3-4}\rangle \to  |0+1\rangle_A |0+1\rangle_B \left( |2\rangle_C + \sum_{t_2=1}^{d_3-d_1} \omega_{d_3-d_1+1}^{s_2 t_2} |t_2+d_1-1\rangle_C \right)|000\rangle_{abc}, \quad  s_2 =1,\dots,d_3-d_1.
		\end{aligned}
		\]
		
		\textbf{Step 2.} Bob performs the measurement
		\[
		\begin{aligned}
			\mathcal{M}_2\equiv \bigl\{
			M_{2,1}=P[(|2\rangle,|d_1\rangle,\dots,|d_2-1\rangle)_B;|0\rangle_{b}],
			\;
			M_{2,2}=I-M_{2,1}
			\bigr\}.
		\end{aligned}
		\]
		If outcome \(M_{2,1}\) occurs, then the subsets are \(\{|\phi_{1+(d_1-1)}\rangle\}\), \(\{|\phi_{m+(d_1-1)}\rangle\}\), and \(\{|\phi_{s_1+d_1+d_2-2}\rangle\}\), which are locally distinguishable. Next, Charlie performs the measurement
		\[
		\begin{aligned}
			\mathcal{M}_3\equiv \bigl\{
			M_{3,1}=P[(|2\rangle,|d_1\rangle,\dots,|d_3-1\rangle)_C;|0\rangle_{c}],
			\;
			M_{3,2}=I-M_{3,1}
			\bigr\}.
		\end{aligned}
		\]
		If outcome \(M_{3,1}\) occurs, then the subsets are \(\{|\phi_{1+d_1+2d_2-2}\rangle\}\), \(\{|\phi_{n+d_1+2d_2-3)}\rangle\}\), and \(\{|\phi_{s_2+d_1+2d_2+d_3-4}\rangle\}\), which are locally distinguishable. Otherwise, the state belongs to one of the remaining possibilities.
		
		\textbf{Step 3.}  Alice then performs the measurement
		\[
		\begin{aligned}
			\mathcal{M}_4\equiv \bigl\{
			M_{4,1}=P[|1\rangle_A;|0\rangle_a],
			\;
			M_{4,2}=I-M_{4,1}
			\bigr\}.
		\end{aligned}
		\]
		If outcome \(M_{4,1}\) occurs, the states are \(|\phi_{1} \rangle\) and $\left |\phi_{1+2d_2-1}  \right  \rangle$, which can be perfectly distinguished by measuring subsystem $B$. Otherwise, we proceed to the next step.
		
		\textbf{Step 4.} Bob performs the measurement
		\[
		\begin{aligned}
			\mathcal{M}_5\equiv \bigl\{
			M_{5,1}=P[|3\rangle_B;(|0\rangle,|1\rangle)_{b}],
			\;
			M_{5,2}=I-M_{5,1}
			\bigr\}.
		\end{aligned}
		\]
		If outcome \(M_{5,1}\) occurs, then the state is \(|\phi_{2+d_1-1}\rangle\). Next, Charlie performs the measurement
		\[
		\begin{aligned}
			\mathcal{M}_6\equiv \bigl\{
			M_{6,1}=P[|3\rangle_C;(|0\rangle,|1\rangle)_{c}],
			\;
			M_{6,2}=I-M_{6,1}
			\bigr\}.
		\end{aligned}
		\]
		If outcome \(M_{6,1}\) occurs, then the state is \(|\phi_{2+d_1+2d_2-2}\rangle\). Otherwise, the state belongs to one of the remaining states.
		
		\textbf{Step 5.} Alice then performs the measurement
		\[
		\begin{aligned}
			\mathcal{M}_7\equiv \bigl\{
			M_{7,1}=P[|2\rangle_A;|1\rangle_a],
			\;
			M_{7,2}=I-M_{7,1}
			\bigr\}.
		\end{aligned}
		\]
		If outcome \(M_{7,1}\) occurs, the states are \(|\phi_{2} \rangle\) and \(|\phi_{2+2d_2-1} \rangle\), which can be perfectly distinguished by measuring subsystem $C$. Otherwise, we proceed to the next step.
		
		\textbf{Step 6.} Charlie performs the measurement
		\[
		\begin{aligned}
			\mathcal{M}_8\equiv \bigl\{
			M_{8,1}=P[|1\rangle_C;(|0\rangle,|2\rangle)_c],
			\;
			M_{8,2}=I-M_{8,1}
			\bigr\}.
		\end{aligned}
		\]
		If outcome \(M_{8,1}\) occurs, then the subsets are \(\{|\phi_{i} \rangle,|\phi_{i+d_1-1} \rangle\}_{i=3}^{d_1-1}\). If outcome \(M_{8,2}\) occurs, the subset is \(\{|\phi_{i+2d_2-1} \rangle\}_{i=3}^{d_1-1}\bigcup\{|\phi_{i+d_1+2d_2-2} \rangle\}_{i=3}^{d_1-2}\). Bob then performs the measurement
		\[
		\begin{aligned}
			\mathcal{M}_9\equiv \bigl\{
			M_{9,1}=P[|1\rangle_A;(|0\rangle,|2\rangle)_a],
			\;
			M_{9,2}=I-M_{9,1}
			\bigr\}.
		\end{aligned}
		\]
		If outcome \(M_{9,1}\) occurs, the state is \(|\phi_{(d_1-1)+(d_1-1)} \rangle\). For outcome \(M_{9,2}\), the subset is \(\{|\phi_{i} \rangle\}_{i=3}^{d_1-1}\bigcup\{|\phi_{i+d_1-1} \rangle\}_{i=3}^{d_1-2}\). It is easy to know that these subsets \(\{|\phi_{i+2d_2-1} \rangle\}_{i=3}^{d_1-1}\bigcup\{|\phi_{i+d_1+2d_2-2} \rangle\}_{i=3}^{d_1-2}\) and \(\{|\phi_{i} \rangle\}_{i=3}^{d_1-1}\bigcup\{|\phi_{i+d_1-1} \rangle\}_{i=3}^{d_1-2}\) are locally distinguishable based on the proof of Theorem \ref{thm:n2}.
		
		If any other outcomes occur in Step 1, similar protocols can be constructed to perfectly distinguish the corresponding states by LOCC.
		
		\section{The proof of Theorem 8}\label{E}
		
		Using the entanglement resource \(|\phi^{+}(2)\rangle_{a_{n-1}a_n}\), the initial state is
		\[
		\begin{aligned}
			\left|\phi\right\rangle_{A_1\cdots A_{n-1}A_n} \otimes \left|\phi^{+}(2)\right\rangle_{a_{n-1}a_n}.
		\end{aligned}
		\]The specific process is as follows.
		
		\textbf{Step 1.} $A_n$ performs the measurement
		\[
		\begin{aligned}
			\mathcal{M}_1 \equiv \bigl\{
			M_{1,1}=P\bigl[(|0\rangle,|2\rangle,|3\rangle,|4\rangle,\dots,|d_1-1\rangle)_{A_{n}};|0\rangle_{a_{n}}\bigr]+P\bigl[|1\rangle_{A_{n}};|1\rangle_{a_{n}}\bigr], \;M_{1,2}=I-M_{1,1}\bigr\}.
		\end{aligned}
		\]
		Suppose that the outcome corresponding to \(M_{1,1}\) occurs. Then
		\[
		\begin{aligned}
			&G_1 \to \{|x_{j_1}\rangle_{A_1}|y_{j_1}\rangle_{A_2}|0\rangle_{A_3}|0\rangle_{A_4}\cdots |0\rangle_{A_{n-3}}|0\rangle_{A_{n-2}}|0\rangle_{A_{n-1}}|1\rangle_{A_n}|11\rangle_{a_{n-1}a_n}\bigr\},\\
			&G_2 \to \{|1\rangle_{A_1}|x_{j_2}\rangle_{A_2}|z_{j_2}\rangle_{A_3}|0\rangle_{A_4}\cdots |0\rangle_{A_{n-3}}|0\rangle_{A_{n-2}}|0\rangle_{A_{n-1}}|0\rangle_{A_n}|00\rangle_{a_{n-1}a_n}\}, \\
			&G_3 \to \{|0\rangle_{A_1}|x_{j_3}\rangle_{A_2}|1\rangle_{A_3}|z_{j_3}\rangle_{A_4}\cdots |0\rangle_{A_{n-3}}|0\rangle_{A_{n-2}}|0\rangle_{A_{n-1}}|0\rangle_{A_n}|00\rangle_{a_{n-1}a_n}\}, \\
			&\vdots \\
			&G_{n-2} \to\{ |0\rangle_{A_1}|x_{j_{n-2}}\rangle_{A_2}|0\rangle_{A_3}|0\rangle_{A_4}\cdots |0\rangle_{A_{n-3}}|1\rangle_{A_{n-2}}|z_{j_{n-2}}\rangle_{A_{n-1}}|0\rangle_{A_n}|00\rangle_{a_{n-1}a_n}\}, \\
			&G_{n-1} \to\{|0\rangle_{A_1}|x_{j_{n-1}}\rangle_{A_2}|0\rangle_{A_3}|0\rangle_{A_4}\cdots |0\rangle_{A_{n-3}}|0\rangle_{A_{n-2}}|1\rangle_{A_{n-1}}|z_{j_{n-1}}'\rangle_{A_na_{n-1}a_n}\}.
		\end{aligned}
		\]Here  \(\{|z_{j_{n-1}}'\rangle_{A_na_{n-1}a_n}\}\) stands for the post-measurement state on corresponding subsystems of \(G_{n-1}\otimes \left|\phi^{+}(2)\right\rangle_{a_{n-1}a_n}\) after we add the ancillary particles.
		
		\textbf{Step 2.} $A_{n-1}$ performs the measurement
		\[
		\begin{aligned}
			\mathcal{M}_2\equiv \bigl\{
			M_{2,1}= P[|0\rangle_{A_{n-1}};|1\rangle_{a_{n-1}}],
			\;
			M_{2,2}=I-M_{2,1}
			\bigr\}.
		\end{aligned}
		\]
		If outcome \(M_{2,1}\) occurs, then the subset is \(G_1\). If outcome \(M_{2,2}\) occurs, we continue.
		
		\textbf{Step 3.}  $A_1$ performs the measurement
		\[
		\begin{aligned}
			\mathcal{M}_3\equiv \bigl\{
			M_{3,1}=|1\rangle_{A_1}\langle 1|,
			\;
			M_{3,2}=I-M_{3,1}
			\bigr\}.
		\end{aligned}
		\]
		If outcome \(M_{3,1}\) occurs, then the subsets are \(G_2\). Then \(A_3,A_4,\dots,A_{n-2}\) perform the measurement  $\mathcal{M}_3 $, respectively, the whole set $\{G_{3},\dots,G_{n-1}\}$ can be distinguished one by one.
		
		The subset \(G_k\) ($k=1,\ldots,n-1$) can be perfectly distinguished by applying the same discrimination protocol as described in Appendix~\ref{A}.
		
		If other outcome occurs in Step 1, similar protocol can be constructed to perfectly distinguish the corresponding subsets.
	\end{appendix}

\end{document}